\documentclass[journal]{IEEEtran}
\usepackage{latexsym,amsmath,amssymb}
\usepackage{supertabular}
\usepackage{longtable}

\newtheorem{thm}{Theorem}[section]
\newtheorem{defn}[thm]{Definition}
\newtheorem{prop}[thm]{Proposition}

\numberwithin{equation}{section}
\newtheorem{rem}[thm]{Remark}
\newtheorem{ex}[thm]{Example}
\newcommand{\F}{\mathbb F}

\def\wt{\mathop{{\rm wt}}}
\def\wth{\textnormal{{\rm wt}}}

\def\tr{\mathop{{\rm tr}}}
\def\H{\mathop{{\rm H}}}
\def\a{\mathop{{\rm a}}}
\def\dist{\mathop{{\rm dist}}}
\def\min{\mathop{{\rm min}}}
\def\F{{\mathbb F}}
\def\u{{\mathbf{u}}}
\def\v{{\mathbf{v}}}
\def\c{{\mathbf{c}}}
\def\d{{\mathbf{d}}}
\def\t{{\mathbf{t}}}
\def\a{{\mathbf{a}}}
\def\b{{\mathbf{b}}}
\def\0{{\mathbf{0}}}
\def\1{{\mathbf{1}}}
\numberwithin{equation}{section}
\begin{document}

\title{Additive Asymmetric Quantum Codes}

\author{
        Martianus~Frederic~Ezerman,~\IEEEmembership{Student Member,~IEEE}, \linebreak[3]%
        San~Ling,~\IEEEmembership{}\linebreak[3]%
        and~Patrick~Sol{\'e},~\IEEEmembership{Member,~IEEE}
\thanks{M. F. Ezerman and S. Ling are with the Division of Mathematical Sciences,
  School of Physical and Mathematical Sciences, Nanyang Technological
  University, SPMS-04-01, 21 Nanyang Link, Singapore 637371, 
  Republic of Singapore 
  (emails:\{mart0005,lingsan\}@ntu.edu.sg).}%
\thanks{P. Sol{\'e} is with the Centre National de la Recherche
  Scientifique (CNRS/LTCI), Telecom-ParisTech, Dept Comelec, 46 rue
  Barrault, 75 634 Paris, France (email: sole@enst.fr).}%
\thanks{The work of the second and third authors was partially supported by 
Singapore National Research Foundation Competitive Research 
Program grant NRF-CRP2-2007-03 and by the Merlion Programme 01.01.06.}
}

% The paper headers
\markboth{submitted to IEEE Transactions on Information Theory, revised on March 15, 2011}%
{Ezerman \MakeLowercase{\textit{et al.}}: Additive Asymmetric Quantum Codes}
% make the title area
\maketitle

\begin{abstract}
We present a general construction of asymmetric quantum codes based 
on additive codes under the trace Hermitian inner product. Various 
families of additive codes over $\F_{4}$ are used in the construction 
of many asymmetric quantum codes over $\F_{4}$.
\end{abstract}

\begin{keywords} Additive codes, BCH codes, circulant codes, 
4-circulant codes, extremal codes, MacDonald codes, nested codes, 
quantum codes, self-orthogonal codes, quantum Singleton bound.
\end{keywords}

\section{Introduction}\label{sec:Intro}
Previously, most of the works on quantum error-correcting codes 
were done with the assumption that the channel is symmetric. 
That is, the various error types were taken to be equiprobable. 
To be brief, the term \textit{quantum codes} or QECC is henceforth 
used to refer to quantum error-correcting codes.

Recently, it has been established that, 
in many quantum mechanical systems, the phase-flip errors 
happen more frequently than the bit-flip errors or the combined 
bit-phase flip errors. For more details, ~\cite{SRK09} can 
be consulted.

There is a need to design quantum codes that take advantage of this 
asymmetry in quantum channels. We call such codes 
\textit{asymmetric quantum codes}. We require the codes to 
correct many phase-flip errors but not necessarily the same number of 
bit-flip errors.

In this paper we extend the construction of asymmetric quantum 
codes in~\cite{WFLX09} to include codes derived from classical 
additive codes under the trace Hermitian inner product.

This work is organized as follows. 
In Section~\ref{sec:Prelims}, we state some basic definitions and 
properties of linear and additive codes. Section ~\ref{sec:QuantumCodes} 
provides an introduction to quantum error-correcting codes in general, 
differentiating the symmetric and the asymmetric cases. 
In Section~\ref{sec:AsymQECC}, a construction of asymmetric QECC based 
on additive codes is presented.

The rest of the paper focuses on additive codes over $\F_{4}$. 
Section~\ref{sec:AdditiveF4} recalls briefly important known facts 
regarding these codes. A construction of asymmetric 
QECC from extremal or optimal self-dual additive codes is given 
in Section~\ref{sec:ExtremalSD}. A construction from 
Hermitian self-orthogonal $\F_{4}$-linear codes is the topic of 
Section~\ref{sec:HSelfOrtho}. Sections~\ref{sec:Cyclic} 
and~\ref{sec:BCHCodes} use nested $\F_{4}$-linear cyclic codes for lengths 
$n \leq 25$ and nested BCH codes for lengths $27 \leq n \leq 51$, 
respectively, in the construction. New or better asymmetric quantum 
codes constructed from nested additive codes over $\F_{4}$ are presented 
in Section~\ref{sec:nestedadditive}, exhibiting the gain of extending the 
construction to include additive codes. Section~\ref{sec:Conclusion} 
provides conclusions and some open problems.

\section{Preliminaries}\label{sec:Prelims}
Let $p$ be a prime and $q = p^f$ for some positive integer $f$. 
An $[n,k,d]_q$-linear code $C$ of length $n$, dimension $k$, 
and minimum distance $d$ is a subspace of dimension $k$ of the 
vector space $\F_q^n$ over the finite field $\F_q=GF(q)$ with 
$q$ elements. For a general, not necessarily linear, code $C$, 
the notation $(n,M=|C|,d)_q$ is commonly used.

The \textit{Hamming weight} of a vector or a codeword $\v$ in a code 
$C$, denoted by $\wt_H(\v)$, is the number of its nonzero entries. 
Given two elements $\u,\v \in C$, the number of positions where 
their respective entries disagree, written as $\dist_H(\u,\v)$, 
is called the \textit{Hamming distance} of $\u$ and $\v$. For any 
code $C$, the \textit{minimum distance} $d = d(C)$ is given by 
$d = d(C) = \min\left\lbrace \dist_H(\u,\v): 
\u,\v \in C,\u \neq \v\right\rbrace$. If $C$ is linear, then its 
closure property implies that $d(C)$ is given by the minimum Hamming 
weight of nonzero vectors in $C$.

We follow~\cite{NRS06} in defining the following three families 
of codes according to their duality types.

\begin{defn}\label{def1.1}
Let $q=r^2=p^f$ be an even power of an arbitrary prime $p$ with 
$\overline{x}=x^{r}$ for $x \in \F_{q}$. Let $n$ be a positive 
integer and $\u = (u_1,\ldots,u_n), \v = (v_1,\ldots,v_n) \in \F_{q}^n$.
\begin{enumerate}
\item $\mathbf{q^{\H}}$ is the family of $\F_{q}$-linear codes of 
length $n$ with the \textit{Hermitian inner product} 
\begin{equation}\label{eq:1.1}
 \left\langle \u,\v\right\rangle _{\H} := 
\sum_{i=1}^{n} u_i \cdot v_i^{\sqrt{q}} \text{.}
\end{equation}

\item $\mathbf{q^{\H+}}$ \textbf{(even)} is the family of trace 
Hermitian codes over $\F_{q}$ of length $n$ which are 
$\F_{r}$-linear, where $r^2=q$ is even. The 
duality is defined according to the \textit{trace Hermitian 
inner product} 
\begin{equation}\label{eq:1.2}
 \left\langle \u,\v\right\rangle _{\tr} := 
\sum_{i=1}^{n} (u_i\cdot v_i^{\sqrt{q}} + u_i^{\sqrt{q}} 
\cdot v_i) \text{.}
\end{equation}

\item $\mathbf{q^{\H+}}$ \textbf{(odd)} is the family of trace 
Hermitian codes over $\F_{q}$ of length $n$ which are 
$\F_{r}$-linear, where $r^2=q$ is odd. The 
duality is defined according to the following inner product, 
which we will still call \textit{trace Hermitian inner product},
%consider giving it a different name
\begin{equation}\label{eq:1.2a}
 \left\langle \u,\v\right\rangle _{\tr} := 
\alpha \cdot \sum_{i=1}^{n} (u_i\cdot v_i^{\sqrt{q}} - u_i^{\sqrt{q}} 
\cdot v_i) \text{,}
\end{equation}
where $\alpha \in \F_{q} \setminus \left\lbrace 0 \right\rbrace$ with 
$\alpha^{r}= -\alpha$.
\end{enumerate}
\end{defn}

\begin{defn}\label{def:1.1a}
A code $C$ of length $n$ is said to be a 
\textit{(classical) additive code} if $C$ belongs to either the 
family $q^{\H+}$ (even) or to the family $q^{\H+}$ (odd).
\end{defn}

Let $C$ be a code. Under a chosen inner product $*$, the 
\textit{dual code} $C^{\perp_{*}}$ of $C$ is given by
\begin{equation*}
 C^{\perp_{*}} := \left\lbrace \u \in \F_q^n : \left\langle 
\u,\v\right\rangle _{*} = 0 \text{ for all } \v \in C 
\right\rbrace \text{.}
\end{equation*}

Accordingly, for a code $C$ in the family $(q^{\H})$,
\begin{equation*}
 C^{\perp_{\H}} := \left\lbrace \u \in \F_q^n : \left\langle 
\u,\v\right\rangle _{\H} = 0 \text{ for all } \v \in C 
\right\rbrace \text{,}
\end{equation*}
and, for a code $C$ in the family $(q^{\H+})$ (even) or $(q^{\H+})$ (odd), 
\begin{equation*}
 C^{\perp_{\tr}} := \left\lbrace \u \in \F_q^n : \left\langle 
\u,\v\right\rangle _{\tr} = 0 \text{ for all } \v \in C 
\right\rbrace \text{.}
\end{equation*}

A code is said to be \textit{self-orthogonal} 
if it is contained in its dual and is said to be \textit{self-dual} 
if its dual is itself. We say that a family of codes is 
\textit{closed} if $(C^{\perp_{*}})^{\perp_{*}} = C$ for each $C$ 
in that family. It has been established~\cite[Ch. 3]{NRS06} that the 
three families of codes in Definition~\ref{def1.1} are closed.

The weight distribution of a code and that of 
its dual are important in the studies of their properties.

\begin{defn}\label{def1.2}
The \textit{weight enumerator} $W_C(X,Y)$ of an $(n,M=|C|,d)_q$-code 
$C$ is the polynomial
\begin{equation}\label{WE}
W_C(X,Y)=\sum_{i=0}^n A_{i} X^{n-i}Y^{i} \text{,}
\end{equation}
where $A_{i}$ is the number of codewords of weight $i$ in the code $C$.
\end{defn}

The weight enumerator of the Hermitian dual code $C^{\perp_{\H}}$ of an 
$[n,k,d]_q$-code $C$ is connected to the weight enumerator 
of the code $C$ via the MacWilliams Equation
\begin{equation}\label{eq:MacW}
 W_{C^{\perp_{\H}}} (X,Y)= \frac{1}{|C|} W_C(X+(q-1)Y,X-Y) \text{.}
\end{equation}

In the case of nonlinear codes, we can define a similar notion 
called the \textit{distance distribution}. The MacWilliams Equation can 
be generalized to the nonlinear cases as well (see~\cite[Ch. 5]{MS77}). 
From~\cite[Sect. 2.3]{NRS06} we know that the families $q^{\H+}$ (even) and 
$q^{\H+}$ (odd) have the same MacWilliams Equation as the family $q^{\H}$. 
Thus, Equation (\ref{eq:MacW}) applies to all three families.

Classical codes are connected to many other combinatorial 
structures. One such structure is the orthogonal array.

\begin{defn}\label{def1.3}
Let $S$ be a set of $q$ symbols or levels. An orthogonal array 
$A$ with $M$ runs, $n$ factors, $q$ levels and strength $t$ 
with index $\lambda$, denoted by $OA(M,n,q,t)$, is an 
$M \times n$ array $A$ with entries from $S$ such that every 
$M \times t$ subarray of $A$ contains each $t$-tuple of $S^t$ 
exactly $\lambda = \frac{M}{q^t}$ times as a row.
\end{defn}

The parameter $\lambda$ is usually not written explicitly 
in the notation since its value depends on $M,q$ and $t$. 
The rows of an orthogonal array are distinct since the 
purpose of its construction is to minimize the number of 
runs in the experiment while keeping some required 
conditions satisfied.

There is a natural correspondence between codes and orthogonal arrays. 
The codewords in a code $C$ can be seen as the rows of an orthogonal 
array $A$ and vice versa. The following proposition due to Delsarte 
(see~\cite[Th. 4.5]{Del73}) will be useful in the sequel. 
Note that the code $C$ in the proposition is a general code. 
No linearity is required. The duality here is defined over any 
inner product. For more on how the dual distance is 
defined for nonlinear codes, we refer to~\cite[Sec. 4.4]{HSS99}.

\begin{prop}\cite[Th. 4.9]{HSS99}\label{OA}
If $C$ is an $(n,M=|C|,d)_q$ code with dual distance $d^{\perp}$, 
then the corresponding orthogonal array is an $OA(M,n,q,d^{\perp}-1)$. 
Conversely, the code corresponding to an $OA(M,n,q,t)$ is an 
$(n,M,d)_q$ code with dual distance $d^{\perp} \geq t+1$. 
If the orthogonal array has strength $t$ but not $t+1$, 
then $d^{\perp}$ is precisely $t+1$.
\end{prop}

\section{Quantum Codes}\label{sec:QuantumCodes}
We assume that the reader is familiar with the standard error model 
in quantum error-correction. The essentials can be found, for 
instance, in~\cite{AK01} and in~\cite{FLX06}. For convenience, 
some basic definitions and results are reproduced here.

Let $\mathbb{C}$ be the field of complex numbers and 
$\eta=e^{\frac{2\pi \sqrt{-1}}{p}}\in \mathbb{C}$. 
We fix an orthonormal basis of 
$\mathbb{C}^{q}$ $$\left\{|v\rangle:v\in \mathbb{F}_{q}\right\}$$ 
with respect to the Hermitian inner product. 
For a positive integer $n$, let 
$V_{n}=(\mathbb{C}^{q})^{\otimes n }$ 
be the $n$-fold tensor product of $\mathbb{C}^{q}$. 
Then $V_{n}$ has the following orthonormal basis
\begin{equation}\label{basis}
\left\{|\c\rangle =|c_{1}c_{2}\ldots c_{n}\rangle
: \c=(c_{1},\ldots,c_{n})  
\in \mathbb{F}_{q}^n\right\} \text{,}
\end{equation}
where $|c_{1}c_{2}\ldots c_{n}\rangle$ abbreviates 
$|c_{1}\rangle\otimes|c_{2}\rangle\otimes\cdots \otimes
|c_{n}\rangle$.

For two quantum states $|\u\rangle$ and $|\v\rangle$ in 
$V_{n}$ with $$|\u\rangle=\sum\limits_{\c\in 
\mathbb{F}_{q}^{n}}\alpha(\c)|\c\rangle,\quad |\v\rangle
=\sum\limits_{\c\in \mathbb{F}_{q}^{n}}\beta(\c)|\c\rangle
\quad (\alpha(\c),\beta(\c)\in \mathbb{C}),$$ 
the Hermitian inner product of $|\u\rangle$ and $|\v\rangle$ is
$$\langle \u|\v\rangle=\sum\limits_{\c\in \mathbb{F}_{q}^{n}}
\overline{\alpha(\c)}\beta(\c)\in \mathbb{C},$$
where  $\overline{\alpha(\c)}$ is the complex conjugate of 
$\alpha(\c)$. We say $|\u\rangle$ and $|\v\rangle$ are 
\textit{orthogonal} if $\langle \u|\v\rangle=0$.

A quantum error acting on $V_{n}$ is a unitary linear operator 
on $V_{n}$ and has the following form
$$e=X(\a)Z(\b)$$ with $\a=(a_{1},\ldots,a_{n}),\b=(b_{1},
\ldots,b_{n})\in \mathbb{F}_{q}^{n}$.

The action of $e$ on the basis (\ref{basis}) of $V_{n}$ is 
$$e|\c\rangle=X(a_{1})Z(b_{1})|c_{1}\rangle
\otimes \ldots \otimes X(a_{n})Z(b_{n})|c_{n}\rangle, $$ where
$$X(a_{i})|c_{i}\rangle=|a_{i}+c_{i}\rangle,
\quad Z(b_{i})|c_{i}\rangle=\eta^{T(b_{i}c_{i})}|c_{i}\rangle$$
with $T:\;\mathbb{F}_{q}\to\mathbb{F}_{p}$ being the trace mapping
$$T(\alpha)=\alpha+\alpha^{p}+\alpha^{p^{2}}+\ldots+\alpha^{p^{m-1}},$$ 
for $q=p^{m}$.
Therefore,
$$e|\c\rangle=\eta^{T(\b\cdot \c)}|\a+\c\rangle,$$
where $\b\cdot \c=\sum\limits_{i=1}^{n}b_{i}c_{i}\in \mathbb{F}_{q}$
is the usual inner product in $\mathbb{F}_{q}^n$.

For $e=X(\a)Z(\b)$ and
$e^{'}=X(\a^{'})Z(\b^{'})$ with $\a,\b$, and $\a^{'},\b^{'}\in
\mathbb{F}_{q}^{n}$,
$$e e^{'}=\eta^{T(\a\cdot \b^{'}-\a^{'}\cdot \b)}e^{'} e.$$
Hence, the set
$$E_{n}=\left\{\eta^{\lambda}X(\a)Z(\b) | 0\leq \lambda\leq p-1, 
\a,\b\in \mathbb{F}_{q}^{n}  \right\}$$
forms a (nonabelian) group, called the \textit{error group} on $V_{n}$.

\begin{defn}\label{def2.1}
For a quantum error $e=\eta^{\lambda}X(\a)Z(\b)\in E_{n}$, we 
define the {\it quantum weight} $w_{Q}(e)$, the {\it $X$-weight} 
$w_{X}(e)$ and the {\it $Z$-weight} $w_{Z}(e)$ of $e$ by
\[
\begin{array}{rcl}
w_{Q}(e)&=&|\{i: 1\leq i\leq n, (a_{i},b_{i})\neq(0,0)\}| \text{,}\\
w_{X}(e)&=&|\{i: 1\leq i\leq n, a_{i}\neq0\}| \text{,}\\
w_{Z}(e)&=&|\{i: 1\leq i\leq n, b_{i}\neq0\}| \text{.}
\end{array}
\]
\end{defn}

Thus, $w_{Q}(e)$ is the number of qudits where 
the action of $e$ is nontrivial by $X(a_{i})Z(b_{i})\not=I$ 
(identity) while $w_{X}(e)$ and $w_{Z}(e)$ are, respectively, 
the numbers of qudits where the $X$-action and 
the $Z$-action of $e$ are nontrivial. We are now ready 
to define the distinction between symmetric and 
asymmetric quantum codes.

\begin{defn}\label{def2.2}
A \textit{$q$-ary  quantum code} of length $n$ is a subspace $Q$ of
$V_{n}$ with dimension $K\geq1$. A quantum code $Q$ of dimension
$K\geq2$ is said to detect $d-1$ qudits of errors for
$d\geq1$ if, for every orthogonal pair $|\u\rangle$, $|\v\rangle$ in
$Q$ with $\langle \u|\v \rangle=0$ and every $e \in E_{n}$ with
$w_{Q}(e)\leq d-1$, $|\u\rangle$ and $e|\v\rangle$ are orthogonal.
In this case, we call $Q$  a \textit{symmetric}
quantum code with parameters $((n,K,d))_{q}$ or
$[[n,k,d]]_{q}$, where $k=\log_{q}K$. Such a quantum code is called
\textit{pure} if $\langle \u|e|\v \rangle=0$ for any $|\u\rangle$ and
$|\v\rangle$ in $Q$ and any $e \in E_{n}$ with $1\leq w_{Q}(e)\leq
d-1$. A quantum code $Q$ with $K=1$ is assumed to be pure.

Let $d_{x}$ and $d_{z}$ be positive integers. A
quantum code $Q$ in $V_{n}$ with dimension $K\geq2$ is called an 
\textit{asymmetric quantum code} with parameters $((n,K,d_{z}/d_{x}))_{q}$
or $[[n,k,d_{z}/d_{x}]]_{q}$, where $k=\log_{q}K$, if $Q$ detects $d_{x}-1$ 
qudits of $X$-errors and, at the same time, $d_{z}-1$ qudits of $Z$-errors. 
That is, if $\langle \u|\v \rangle=0$ for $|\u\rangle,|\v\rangle\in Q$, 
then $\langle \u|e|\v \rangle=0$ for any
$e \in E_{n}$ such that $w_{X}(e)\leq d_{x}-1$ and $w_{Z}(e)\leq
d_{z}-1$. Such an asymmetric quantum code $Q$ is called \textit{pure}
if $\langle \u|e|\v \rangle=0$ for any $|\u\rangle,|\v\rangle\in Q$ and
$e \in E_{n}$ such that $1 \leq w_{X}(e)\leq d_{x}-1$ and $1 \leq w_{Z}(e)\leq
d_{z}-1$. An asymmetric quantum code $Q$ with $K=1$ is assumed to be pure.
\end{defn}

\begin{rem}\label{rem2.3}
An asymmetric quantum code with parameters  $((n,K,d/d))_{q}$ 
is a symmetric quantum code with parameters $((n,K,d))_{q}$, 
but the converse is not true since, for $e\in E_{n}$ with 
$w_{X}(e)\leq d-1$ and $w_{Z}(e)\leq d-1$, the weight 
$w_{Q}(e)$ may be bigger than $d-1$.
\end{rem}

Given any two codes $C$ and $D$, let the notation 
$\wt(C \setminus D)$ denote $\min\left\lbrace \wt_{H}(\u (\neq \0)) : 
\u\in(C \setminus D)\right\rbrace$. The analogue of the well-known CSS 
construction (see~\cite{CRSS98}) for the asymmetric case is known.

\begin{prop}\cite[Lemma 3.1]{SRK09}\label{prop2.4}
Let $C_x,C_z$ be linear codes over $\F_q^n$ with parameters 
$[n,k_x]_q$, and $[n,k_z]_q$ respectively. Let 
$C_x^\perp \subseteq C_z$. Then there exists an 
$[[n,k_x +k_z-n,d_{z} /d_{x}]]_q$ asymmetric quantum code, 
where $d_x =\wt(C_x \setminus C_z^\perp)$ and 
$d_z =\wt(C_z \setminus C_x^\perp)$.
\end{prop}

The resulting code is said to be \textit{pure} if, 
in the above construction, $d_x=d(C_x)$ and $d_z=d(C_z)$.

\section{Asymmetric QECC from Additive Codes}\label{sec:AsymQECC}
The following result has been established recently:

\begin{thm}\cite[Th. 3.1]{WFLX09}\label{thm:3.1}
\begin{enumerate}
 \item There exists an asymmetric quantum code with parameters 
$((n,K,d_z/d_x))_q$ with $K \geq 2$ if and only if there exist 
$K$ nonzero mappings
\begin{equation}\label{eq:3.1}
\varphi_i : \F_q^n \rightarrow \mathbb{C} \text{ for } 1\leq i \leq K 
\end{equation}

satisfying the following conditions: for each $d$ such that 
$1 \leq d \leq \min\left\lbrace d_x,d_z\right\rbrace $ and 
partition of $\left\lbrace 1,2,\ldots,n\right\rbrace $,
\begin{equation}\label{eq:3.2}
\begin{cases}
\left\lbrace 1,2,\ldots,n \right\rbrace = A \cup X \cup Z \cup B \text{,} \\
|A| = d-1,\quad |B| = n+d-d_x-d_z+1 \text{,}\\
|X|=d_x - d,\quad |Z| = d_z - d \text{,} 
\end{cases}
\end{equation}

and each $\c_A,\c_A' \in \F_q^{|A|}$, 
$\c_Z \in \F_q^{|Z|}$ and $\a_X \in \F_q^{|X|}$, 
we have the equality
\begin{multline}\label{eq:3.3}
\sum_{\substack{ \c_X \in \F_q^{|X|} \text{,}\\ \c_B \in \F_q^{|B|}}} 
\overline{\varphi_i (\c_A,\c_X,\c_Z,\c_B)} \varphi_j(\c_A',\c_X - \a_X,\c_Z,\c_B) \\
= 
\begin{cases}
0 &\text{for $i \neq j$,} \\
I(\c_A,\c_A',\c_Z,\a_X) &\text{for $i = j$,}
\end{cases}
\end{multline}

where $I(\c_A,\c_A',\c_Z,\a_X)$ is an element of $\mathbb{C}$ 
which is independent of $i$. The notation $(\c_A,\c_X,\c_Z,\c_B)$ 
represents the rearrangement of the entries of the 
vector $\c \in \F_q^n$ according to the partition of 
$\left\lbrace 1,2,\ldots,n\right\rbrace $ given in 
Equation (\ref{eq:3.2}).

\item Let $(\varphi_i,\varphi_j)$ stand for 
$\sum_{\c\in \F_q^n} \overline{\varphi_i(\c)}\varphi_j(\c)$.
There exists a pure asymmetric quantum code with 
parameters $((n,K\geq1,d_z/d_x))_q$ 
if and only if there exist $K$ nonzero mappings 
$\varphi_i$ as shown in Equation (\ref{eq:3.1}) such that
\begin{itemize}
 \item $\varphi_i$ are linearly independent for 
$1\leq i\leq K$, i.e., the rank of the $K \times q^n$ matrix
 $(\varphi_i (\c))_{1\leq i\leq K, \c \in \F_q^n}$ is $K$; 
and
 \item for each $d$ with 
$1 \leq d \leq \min\left\lbrace d_x,d_z\right\rbrace $, 
a partition in Equation (\ref{eq:3.2}) 
and $\c_A,\a_A \in \F_q^{|A|}, \c_Z \in \F_q^{|Z|}$ 
and $\a_X \in \F_q^{|X|}$, we have the equality
\end{itemize}
\end{enumerate}
\begin{multline}\label{eq:3.4}
\sum_{\substack{ \c_X \in \F_q^{|X|}, \\ \c_B \in 
\F_q^{|B|}}} \overline{\varphi_i (\c_A,\c_X,\c_Z,\c_B)} 
\varphi_j(\c_A+\a_A,\c_X + \a_X,\c_Z,\c_B) \\
= 
\begin{cases}
0 &\text{for $(\a_A,\a_X) \neq (\0,\0)$,} \\
\frac{(\varphi_i,\varphi_j)}{q^{d_z-1}} 
&\text{for $(\a_A,\a_X) = (\0,\0)$.}
\end{cases}
\end{multline}
\end{thm}

The following result is due to Keqin Feng and Long Wang. It has, 
however, never appeared formally in a published form before. 
Since it will be needed in the sequel, we present it here with 
a proof.
\begin{prop}(K.~Feng and L.~Wang)\label{prop:3.2} Let $a,b$ be 
positive integers. 
There exists an asymmetric quantum code $Q$ with parameters 
$((n,K,a/b))_q$ if and only if there exists an 
asymmetric quantum code $Q'$ with parameters $((n,K,b/a))_q$.
$Q'$ is pure if and only if $Q$ is pure.
\end{prop}

\begin{proof}
We begin by assuming the existence of an $((n,K,a/b))_q$ asymmetric 
quantum code $Q$. Let $\varphi_{i}$ with $1 \leq i \leq K$ be the 
$K$ mappings given in Theorem~\ref{thm:3.1}. Define the following 
$K$ mappings
\begin{equation}\label{eq:3.5}
\begin{aligned}
\Phi_i :\quad &\F_q^n \rightarrow \mathbb{C} 
\text{ for } 1\leq i \leq K \\
&\v \mapsto \sum_{\c \in \F_{q}^{n}} \varphi_{i}(\c) \eta^{T(\c \cdot \v)}.
\end{aligned}
\end{equation}

Let $\v_{A},\b_{A} \in \F_{q}^{|A|}, \v_{X} \in \F_{q}^{|X|}$, and 
$\b_{Z} \in \F_{q}^{|Z|}$.
For each $d$ such that $1 \leq d \leq \min\left\lbrace d_x,d_z\right\rbrace $ 
and a partition of $\left\lbrace 1,2,\ldots,n \right\rbrace $ 
given in Equation (\ref{eq:3.2}), we show that
\begin{multline}\label{eq:3.6}
S = \sum_{\substack{ \v_Z \in \F_q^{|Z|} \text{,}\\ \v_B \in \F_q^{|B|}}} 
\overline{\Phi_i (\v)} \Phi_j(\v_A+\b_A,\v_X,\v_Z+\b_Z,\v_B) \\
= 
\begin{cases}
0 &\text{for $i \neq j$,} \\
I'(\v_A,\b_A,\b_Z,\v_X) &\text{for $i = j$,}
\end{cases}
\end{multline}
where $I'(\v_A,\b_A,\b_Z,\v_X)$ is an element of $\mathbb{C}$ 
which is independent of $i$.

Let $\t=(\v_A+\b_A,\v_X,\v_Z+\b_Z,\v_B)$. Applying Equation (\ref{eq:3.5}) 
yields
\begin{equation}\label{eq:3.7}
 S = \sum_{\substack{ \v_Z \in \F_q^{|Z|} \text{,}\\ \v_B \in \F_q^{|B|}}}
\sum_{\c,\d \in \F_{q}^n} \overline{\varphi_i (\c)} \varphi_j (\d) 
\eta^{T((-\c \cdot \v)+(\d \cdot \t))} \text{.} 
\end{equation}
By carefully rearranging the summations and grouping the terms, we get
\begin{equation}\label{eq:3.8}
 S = \sum_{\c,\d \in \F_{q}^n} \overline{\varphi_i (\c)} \varphi_j (\d)
\cdot \kappa \cdot \lambda \text{,}
\end{equation}
where
\begin{align*}
 \kappa &= \eta^{T(\v_A \cdot(\d_A-\c_A)+\v_X \cdot(\d_X-\c_X)+
\d_A \cdot \b_A+\d_Z \cdot \b_Z)} \text{,} \\
 \lambda &= \sum_{\substack{ \v_Z \in \F_q^{|Z|} \text{,}\\ \v_B \in \F_q^{|B|}}}
\eta^{T(\v_B \cdot(\d_B-\c_B)+\v_Z \cdot(\d_Z-\c_Z))} \text{.}
\end{align*}
By orthogonality of characters,
\begin{equation*}
\lambda =
\begin{cases}
q^{|Z|+|B|} & \text{if }\d_B = \c_B \text{ and } \d_Z= \c_Z \text{,} \\
0 & \text{otherwise.}
\end{cases}
\end{equation*}
Therefore,
\begin{equation}\label{eq:3.9}
 S=\sum_{\substack{\c \in \F_{q}^n \\ \d_A \in \F_q^{|A|} \text{,}\d_X \in \F_q^{|X|}}}
q^{|Z|+|B|} \cdot \overline{\varphi_i (\c)} \varphi_j (\d_A,\d_X,\c_Z,\c_B) \cdot \pi \text{,}
\end{equation}
where
\begin{equation*}
 \pi = \eta^{T(\v_A \cdot(\d_A-\c_A)+\v_X \cdot(\d_X-\c_X)+\d_A \cdot \b_A+\c_Z \cdot \b_Z)}.
\end{equation*}
Now, we let $k=n-d_x+1$, $\a_A=\d_A-\c_A$, and $\a_X=\d_X-\c_X$. 
Splitting up the summation once again yields
\begin{multline}\label{eq:3.10}
 S=q^k \sum_{\substack{\c_A,\a_A \in \F_{q}^{|A|} \\ \c_Z \in \F_q^{|Z|} \text{,}\a_X \in \F_q^{|X|}}}
\eta^{T(\v_A \cdot \a_A+\v_X \cdot\a_X+\b_A \cdot (\c_A + \a_A) +\c_Z \cdot \b_Z)} \\
\cdot \sum_{\substack{ \c_X \in \F_q^{|X|} \text{,}\\ \c_B \in \F_q^{|B|}}} 
\overline{\varphi_i (\c_A,\c_X,\c_Z,\c_B)} \varphi_j(\c_A+\a_A,\c_X + \a_X,\c_Z,\c_B) \text{.}
\end{multline}
Invoking Equation (\ref{eq:3.3}) concludes the proof for the first part 
with $I'$ given by
\begin{equation}\label{eq:3.11}
 I'=q^k I 
\sum_{\substack{\c_A,\a_A \in \F_{q}^{|A|} \\ \c_Z \in \F_q^{|Z|} \text{,}\a_X \in \F_q^{|X|}}}
\eta^{T(\v_A \cdot \a_A+\v_X \cdot\a_X+\b_A \cdot (\c_A + \a_A) +\c_Z \cdot \b_Z)} \text{.}
\end{equation}

For the second part, let us assume the existence of a pure $((n,K,a/b))_q$ 
asymmetric quantum code $Q$. Note that the Fourier transformations $\Phi_i$ for 
$1 \leq i\leq K$ are linearly independent. We use Equations (\ref{eq:3.10}) 
and (\ref{eq:3.4}) to establish the equality
\begin{multline}\label{eq:3.12}
S = \sum_{\substack{ \v_Z \in \F_q^{|Z|} \text{,}\\ \v_B \in \F_q^{|B|}}} 
\overline{\Phi_i (\v)} \Phi_j(\v_A+\b_A,\v_X,\v_Z+\b_Z,\v_B) \\
= 
\begin{cases}
0 &\text{for $(\b_A,\b_Z) \neq (\0,\0)$,} \\
q^{n} \frac{(\varphi_i,\varphi_j)}{q^{d_x-1}} 
&\text{for $(\b_A,\b_Z) = (\0,\0)$.}
\end{cases}
\end{multline}

Consider the term
\begin{equation*}
 M:=\sum_{\substack{ \c_X \in \F_q^{|X|} \text{,}\\ \c_B \in \F_q^{|B|}}} 
\overline{\varphi_i (\c)} \varphi_j(\c_A+\a_A,\c_X + \a_X,\c_Z,\c_B)
\end{equation*}
in Equation (\ref{eq:3.10}). By the purity assumption, for 
$(\a_A,\a_X) \neq (\0,\0)$, $M=0$. 
For $(\a_A,\a_X)=(\0,\0)$, $M=\frac{(\varphi_i,\varphi_j)}{q^{d_{z}-1}}$. 
Hence,
\begin{equation}\label{eq:3.13}
 S=q^k \sum_{\substack{\c_A \in \F_{q}^{|A|} \text{,}\\ \c_Z \in \F_q^{|Z|}}} 
\eta^{T(\b_A \cdot \c_A + \b_Z \cdot \c_Z)}
\cdot \frac{(\varphi_i,\varphi_j)}{q^{d_z-1}} \text{.}
\end{equation}
By orthogonality of characters, if $(\b_A,\b_Z) \neq (\0,\0)$, then
\begin{equation*}
 \sum_{\substack{\c_A \in \F_{q}^{|A|} \text{,}\\ \c_Z \in \F_q^{|Z|}}}
\eta^{T(\b_A \cdot \c_A + \b_Z \cdot \c_Z)} = 0 \text{,}
\end{equation*}
making $S=0$. If $(\b_A,\b_Z)=(\0,\0)$, then
\begin{equation*}
 S = q^k \cdot q^{|A|+|Z|} \cdot \frac{(\varphi_i,\varphi_j)}{q^{d_z-1}} \text{.}
\end{equation*}
This completes the proof of the second part.
\end{proof}
With this result, without loss of generality, $d_z \geq d_x$ is henceforth assumed.

\begin{rem}\label{rem3.3}
If we examine closely the proof of Theorem~\ref{thm:3.1} above 
as presented in Theorem 3.1 of~\cite{WFLX09}, only the additive property 
(instead of linearity) is used. We will show that 
the conclusion of the theorem with an adjusted value for $K$ 
still follows if we use \underline{classical additive codes} 
instead of linear codes.
\end{rem}

\begin{thm}\label{thm:3.4}
Let $d_x$ and $d_z$ be positive integers. Let $C$ be 
a classical additive code in $\F_q^n$. Assume that 
$d^{\perp_{\tr}}= d(C^{\perp_{\tr}})$ is the minimum 
distance of the dual code $C^{\perp_{\tr}}$ of $C$ 
under the trace Hermitian inner product. For a set 
$V:=\left\lbrace \v_i : 1\leq i \leq K\right\rbrace $ 
of $K$ distinct vectors in $\F_q^n$, let
$d_v:=\min\left\lbrace \wt_{H}(\v_i - \v_j + \c)
: 1 \leq i \neq j \leq K, \c \in C\right\rbrace$.
If $d^{\perp_{\tr}} \geq d_z$ and $d_v \geq d_x$, 
then there exists an asymmetric quantum code $Q$ with 
parameters $((n,K,d_z/d_x))_q$.
\end{thm}

\begin{proof}
Define the following functions
\begin{equation}\label{eq:3.14}
\begin{aligned}
\varphi_i :\quad &\F_q^n \rightarrow \mathbb{C} 
\text{ for } 1\leq i \leq K \\
&\u \mapsto
\begin{cases}
1 &\text{if $\u \in \v_i+C$,} \\
0 &\text{if $\u \not \in \v_i+C$.}
\end{cases}
\end{aligned}
\end{equation}

For each $d$ such that $1 \leq d \leq \min\left\lbrace d_x,d_z\right\rbrace $ 
and a partition of $\left\lbrace 1,2,\ldots,n \right\rbrace $ 
given in Equation (\ref{eq:3.2}),
\begin{equation*}
 \overline{\varphi_i (\c_A,\c_X,\c_Z,\c_B)} 
\varphi_j(\c_A+\a_A,\c_X + \a_X,\c_Z,\c_B) \neq 0
\end{equation*}
if and only if
\begin{equation*}
\begin{cases}
(\c_A,\c_X,\c_Z,\c_B) &\in \v_i+C \text{,} \\
(\c_A+\a_A,\c_X +\a_X,\c_Z,\c_B) &\in \v_j+C \text{,}
\end{cases}
\end{equation*}
which, in turn, is equivalent to
\begin{equation}\label{eq:3.15}
\begin{cases}
(\c_A,\c_X,\c_Z,\c_B) &\in \v_i+C \text{,} \\
(\a_A,\a_X,\0_Z,\0_B) &\in \v_j-\v_i+C \text{.}
\end{cases}
\end{equation}

Note that since $\wth_{H}(\a_A,\a_X,\0_Z,\0_B) \leq |A|+|X| = d_x-1$, 
we know that $(\a_A,\a_X,\0_Z,\0_B) \in \v_j-\v_i+C$ 
means $i=j$ by the definition of $d_v$ above.
Thus, if $i \neq j$,
\begin{equation}\label{eq:3.16}
\sum_{\substack{ \c_X \in \F_q^{|X|} \\ \c_B \in \F_q^{|B|}}} 
\overline{\varphi_i (\c_A,\c_X,\c_Z,\c_B)} 
\varphi_j(\c_A+\a_A,\c_X + \a_X,\c_Z,\c_B) = 0 \text{.}
\end{equation}

Now, consider the case of $i = j$. By Equation (\ref{eq:3.15}), 
if $(\a_A,\a_X,\0_Z,\0_B) \not \in C$, then it has no contribution 
to the sum we are interested in. If $(\a_A,\a_X,\0_Z,\0_B) \in C$, 
then
\begin{multline}\label{eq:3.17}
\sum_{\substack{ \c_X \in \F_q^{|X|} \\ \c_B \in \F_q^{|B|}}} 
\overline{\varphi_i (\c_A,\c_X,\c_Z,\c_B)} 
\varphi_i(\c_A+\a_A,\c_X + \a_X,\c_Z,\c_B) \\
= \sum_{\begin{subarray}{c} 
         \c_X \in \F_q^{|X|}, \c_B \in \F_q^{|B|} \\ 
         (\c_A,\c_X,\c_Z,\c_B) \in \v_i+C
        \end{subarray}} 
 1 \text{.}
\end{multline}

Proposition~\ref{OA} above tells us that, if $C$ is any 
classical $q$-ary code of length $n$ and size $M$ such 
that the minimum distance $d^{\perp}$ of its dual is greater 
than or equal to $d_z$, then any coset of $C$ is an 
orthogonal array of level $q$ and of strength exactly 
$d_z-1$. In other words, there are exactly 
$\frac{|C|}{q^{d_z-1}}$ vectors 
$(\c_A,\c_X,\c_Z,\c_B) \in \v_i+C$ for any fixed 
$(\c_A,\c_Z) \in \F_q^{d_z-1}$. Thus, for $i = j$, 
the sum on the right hand side of Equation (\ref{eq:3.17}) 
is $\frac{|C|}{q^{d_z-1}}$, which is independent of $i$. 
By Theorem~\ref{thm:3.1} we have an asymmetric 
quantum code $Q$ with parameters $((n,K,d_z/d_x))_q$.
\end{proof}

\begin{thm}\label{thm:3.5}
Let $q=r^2$ be an even power of a prime $p$. For $i = 1,2$, 
let $C_i$ be a classical additive code with parameters 
$(n,K_i,d_i)_q$. If $C_1^{\perp_{\tr}} \subseteq C_2$, then 
there exists an asymmetric quantum code $Q$ with parameters 
$((n,\frac{|C_2|}{|C_1^{\perp_{\tr}}|},d_z/d_x))_q$ where 
$\left\lbrace d_z,d_x\right\rbrace = \left\lbrace d_1,d_2\right\rbrace$. 
\end{thm}

\begin{proof}
We take $C = C_1^{\perp_{\tr}}$ in Theorem~\ref{thm:3.4} above. 
Since $C_1^{\perp_{\tr}} \subseteq C_2$, we have 
$C_2 = C_1^{\perp_{\tr}} \oplus C'$, where $C'$ is an 
$\F_r$-submodule of $C_2$ and $\oplus$ is the direct sum so that 
$|C'|=\frac{|C_2|}{|C_1^{\perp_{\tr}}|}$. 
Let $C' = \left\lbrace \v_1,\ldots,\v_K\right\rbrace$, 
where $K = \frac{|C_2|}{|C_1^{\perp_{\tr}}|}$. Then
\begin{align*}
 d^{\perp_{\tr}} &= d(C^{\perp_{\tr}}) = d(C_1) = d_1 \text{ and} \\
 d_v &= \min\left\lbrace \wth_{H}(\v_i - \v_j + \c) : 
1 \leq i \neq j \leq K, \c \in C \right\rbrace \\
&= \min \left\lbrace \wth_{H}(\v + \c) : 
\0 \neq \v \in C', \c \in C_1^{\perp_{\tr}} \right\rbrace \geq d_2 \text{.}
\end{align*}
\end{proof}

Theorem~\ref{thm:3.5} can now be used to construct quantum codes. 
In this paper, all computations are done in MAGMA~\cite{BCP97} 
version V2.16-5. 

The construction method of Theorem~\ref{thm:3.5} 
falls into what some have labelled the \textit{CSS-type construction}. 
It is noted in~\cite[Lemma 3.3]{SRK09} that any CSS-type 
$\F_{q}$-linear $[[n,k,d_{z}/d_{x}]]_{q}$-code satisfies 
the quantum version of the Singleton bound
\begin{equation*}
k \leq n-d_{x}-d_{z}+2 \text{.}
\end{equation*}
This bound is conjectured to hold for all asymmetric 
quantum codes. Some of our codes in later sections attain 
$k = n-d_{x}-d_{z}+2$. 
They are printed in boldface throughout the tables and examples.

\section{Additive Codes over $\F_4$}\label{sec:AdditiveF4}

Let $\F_4:=\left\lbrace 0,1,\omega,\omega^{2}=\overline{\omega}\right\rbrace $. 
For $x \in \F_{4}$, $\overline{x}=x^{2}$, the conjugate of $x$. 
By definition, an additive code $C$ of length $n$ over $\F_{4}$ is a 
free $\F_{2}$-module. It has size $2^{l}$ for some $0 \leq l \leq 2n$. 
As an $\F_{2}$-module, $C$ has a basis consisting of $l$ basis vectors. 
A \textit{generator matrix} of $C$ is an $l \times n$ matrix 
with entries elements of $\F_{4}$ whose rows form a basis of $C$.

Additive codes over $\F_4$ equipped with the trace Hermitian 
inner product have been studied primarily in connection to designs 
(e.g.~\cite{KP03}) and to stabilizer quantum codes 
(e.g.~\cite{GHKP01} and~\cite[Sec. 9.10]{HP03}). 
It is well known that if $C$ is an additive $(n,2^l)_4$-code, 
then $C^{\perp_{\tr}}$ is an additive $(n,2^{2n-l})_4$-code.

To compute the weight enumerator of $C^{\perp_{\tr}}$ we use 
Equation (\ref{eq:MacW}) with $q=4$
\begin{equation}\label{eq:4.1}
W_{C^{\perp_{\tr}}}(X,Y) = \frac{1}{|C|} W_C(X+3Y,X-Y) \text{.}
\end{equation}

\begin{rem}\label{rem:4.1}
If the code $C$ is $\F_4$-linear with parameters $[n,k,d]_4$, 
then $C^{\perp_{\H}} = C^{\perp_{\tr}}$. This is because 
$C^{\perp_{\H}}$ is of size $4^{n-k} = 2^{2n-2k}$ which is 
also the size of $C^{\perp_{\tr}}$. Alternatively, one can invoke 
~\cite[Th. 3]{CRSS98}.
\end{rem}

From here on, we assume the trace Hermitian inner product whenever additive 
$\F_{4}$ codes are discussed and the Hermitian inner product whenever 
$\F_{4}$-linear codes are used.

Two additive codes $C_1$ and $C_2$ over $\F_{4}$ 
are said to be \textit{equivalent} if there is a map sending the codewords of 
one code onto the codewords of the other where the map consists of a 
permutation of coordinates, followed by a scaling of coordinates by 
elements of $\F_{4}$, followed by a conjugation of the entries of some of 
the coordinates. 

\section{Construction from Extremal or Optimal Additive Self-Dual Codes over $\F_4$}
\label{sec:ExtremalSD}
As a direct consequence of Theorem~\ref{thm:3.5}, 
we have the following result.

\begin{thm}\label{thm:5.1}
If $C$ is an additive self-dual code of parameters 
$(n,2^n,d)_4$, then there exists an $[[n,0,d_z/d_x]]_4$ asymmetric 
quantum code $Q$ with $d_z = d_x = d(C^{\perp_{\tr}})$.
\end{thm}

Additive self-dual codes over $\F_{4}$ exist for any length $n$ since 
the identity matrix $I_{n}$ clearly generates a self-dual 
$(n,2^{n},1)_{4}$-code. Any linear self-dual $[n,n/2,d]_{4}$-code 
is also an additive self-dual $(n,2^{n},d)_{4}$-code.

\begin{defn}\label{defn:5.2}
A self-dual $(n,2^n,d)_4$-code $C$ is \textit{Type II} if all of its 
codewords have even weight. If $C$ has a codeword of odd weight, then $C$ 
is \textit{Type I}. 
\end{defn}
It is known (see~\cite[Sec. 4.2]{RS98}) that Type II codes of length $n$ 
exist only if $n$ is even and that a Type I code is not $\F_4$-linear. There 
is a bound in~\cite[Th. 33]{RS98} on the minimum weight of an additive 
self-dual code. If $d_{I}$ and $d_{II}$ are the minimum weights of Type I 
and Type II codes of length $n$, respectively, then
\begin{equation}
\begin{aligned}
d_{I} &\leq 
\begin{cases}
2 \lfloor \frac{n}{6} \rfloor + 1 \text{,} \hfil \text{ if } n \equiv 0 \pmod{6} \\
2 \lfloor \frac{n}{6} \rfloor + 3 \text{,} \hfil \text{ if } n \equiv 0 \pmod{6} \\
2 \lfloor \frac{n}{6} \rfloor + 2 \text{,} \hfil \text{otherwise}
\end{cases} \text{,}\\
d_{II} &\leq 2 \lfloor \frac{n}{6} \rfloor + 2 \text{.}
\end{aligned}
\end{equation}

A code that meets the appropriate bound is called \textit{extremal}. If a code is 
not extremal yet no code of the given type can exist with a larger minimum weight, 
then we call the code \textit{optimal}.

The complete classification, up to equivalence, of additive self-dual codes 
over $\F_{4}$ up to $n=12$ can be found in~\cite{DP06}. The classification of 
extremal codes of lengths $n=13$ and $n=14$ is presented in~\cite{Var07}. 
Many examples of good additive codes for larger values of $n$ are presented 
in~\cite{GK04},~\cite{Var07}, and~\cite{Var09}.

Table~\ref{table:SDCodes} summarizes the results thus far and lists down the resulting 
asymmetric quantum codes for lengths up to $n=30$. The subscripts $_{I}$ and $_{II}$ 
indicates the types of the codes. The superscripts $^{e},^{o},^{b}$ indicate 
the fact that the minimum distance $d$ is extremal, optimal, and best-known 
(not necessarily extremal or optimal), respectively. The number of codes for each 
set of given parameters is listed in the column under the heading \textbf{num}.

\begin{table*}
\caption{Best-Known Additive Self-Dual Codes over $\F_{4}$ for $n \leq 30$ and the Resulting Asymmetric Quantum Codes}
\label{table:SDCodes}
\centering
\begin{tabular}{|| c | c | c | c | l | c | c | c | l || c | c | c | c | l ||}
\hline
\textbf{$n$} & \textbf{$d_{I}$} &\textbf{num$_{I}$} & \textbf{Ref.} & \textbf{Code $Q$} & 
\textbf{$d_{II}$} &\textbf{num$_{II}$} & \textbf{Ref.} & \textbf{Code $Q$} &
\textbf{$n$} & \textbf{$d_{I}$} &\textbf{num$_{I}$} & \textbf{Ref.} & \textbf{Code $Q$} \\
\hline
$2$ & $1^{o}$ & 1 & \cite{DP06} & $[[2,0,1/1]]_{4}$ & $2^{e}$ & 1 & \cite{DP06} & $\mathbf{[[2,0,2/2]]_{4}}$ &
$3$ & $2^{e}$ & 1 & \cite{DP06} & $[[3,0,2/2]]_{4}$ \\

$4$ & $2^{e}$ & 1 & \cite{DP06} & $[[4,0,2/2]]_{4}$ & $2^{e}$ & 2 & \cite{DP06} & $[[4,0,2/2]]_{4}$ &
$5$ & $3^{e}$ & 1 & \cite{DP06} & $[[5,0,3/3]]_{4}$ \\

$6$ & $3^{e}$ & 1 & \cite{DP06} & $[[6,0,3/3]]_{4}$ & $4^{e}$ & 1 & \cite{DP06} & $\mathbf{[[6,0,4/4]]_{4}}$ &
$7$ & $3^{o}$ & 4 & \cite{DP06} & $[[7,0,3/3]]_{4}$ \\

$8$ & $4^{e}$ & 2 & \cite{DP06} & $[[8,0,4/4]]_{4}$ & $4^{e}$ & 3 & \cite{DP06} & $[[8,0,4/4]]_{4}$ &
$9$ & $4^{e}$ & 8 & \cite{DP06} & $[[9,0,4/4]]_{4}$ \\

$10$ & $4^{e}$ & 101 & \cite{DP06} & $[[10,0,4/4]]_{4}$ & $4^{e}$ & 19 & \cite{DP06} & $[[10,0,4/4]]_{4}$ &
$11$ & $5^{e}$ & 1 & \cite{DP06} & $[[11,0,5/5]]_{4}$\\

$12$ & $5^{e}$ & 63 & \cite{DP06} & $[[12,0,5/5]]_{4}$ & $6^{e}$ & 1 & \cite{DP06} & $[[12,0,6/6]]_{4}$ &
$13$ & $5^{o}$ & 85845 & \cite{Var07} & $[[13,0,5/5]]_{4}$ \\

$14$ & $6^{e}$ & 2 & \cite{Var07} & $[[14,0,6/6]]_{4}$ & $6^{e}$ & 1020 & \cite{DP06} & $[[14,0,6/6]]_{4}$ &
$15$ & $6^{e}$ & $\geq 2118$ & \cite{Var07} & $[[15,0,6/6]]_{4}$ \\

$16$ & $6^{e}$ & $\geq 8369$ & \cite{Var07} & $[[16,0,6/6]]_{4}$ & $6^{e}$ & $\geq 112$ & \cite{Var07} & $[[16,0,6/6]]_{4}$ &
$17$ & $7^{e}$ & $\geq 2$ & \cite{Var07} & $[[17,0,7/7]]_{4}$ \\

$18$ & $7^{e}$ & $\geq 2$ & \cite{Var07} & $[[18,0,7/7]]_{4}$ & $8^{e}$ & $\geq 1$ & \cite{Var07} & $[[18,0,8/8]]_{4}$ &
$19$ & $7^{b}$ & $\geq 17$ & \cite{Var07} & $[[19,0,7/7]]_{4}$ \\

$20$ & $8^{e}$ & $\geq 3$ & \cite{GK04} & $[[20,0,8/8]]_{4}$ & $8^{e}$ & $\geq 5$ & \cite{GK04} & $[[20,0,8/8]]_{4}$ &
$21$ & $8^{e}$ & $\geq 2$ & \cite{Var07} & $[[21,0,8/8]]_{4}$ \\

$22$ & $8^{e}$ & $\geq 1$ & \cite{GK04} & $[[22,0,8/8]]_{4}$ & $8^{e}$ & $\geq 67$ & \cite{GK04} & $[[22,0,8/8]]_{4}$ &
$23$ & $8^{b}$ & $\geq 2$ & \cite{GK04} & $[[23,0,8/8]]_{4}$ \\

$24$ & $8^{b}$ & $\geq 5$ & \cite{Var09} & $[[24,0,8/8]]_{4}$ & $8^{b}$ & $\geq 51$ & \cite{GK04} & $[[24,0,8/8]]_{4}$ &
$25$ & $8^{b}$ & $\geq 30$ & \cite{GK04} & $[[25,0,8/8]]_{4}$ \\

$26$ & $8^{b}$ & $\geq 49$ & \cite{Var09} & $[[26,0,8/8]]_{4}$ & $8^{b}$ & $\geq 161$ & \cite{Var09} & $[[26,0,8/8]]_{4}$ &
$27$ & $8^{b}$ & $\geq 15$ & \cite{GK04} & $[[27,0,9/9]]_{4}$ \\

$28$ & $10$ & ? & \cite{GK04} & & $10^{e}$ & $\geq 1$ & \cite{GK04} & $[[28,0,10/10]]_{4}$ &
$29$ & $11^{e}$ & $\geq 1$ & \cite{GK04} & $[[29,0,11/11]]_{4}$ \\

$30$ & $11$ & ? & \cite{GK04} & & $12^{e}$ & $\geq 1$ & \cite{GK04} & $[[30,0,12/12]]_{4}$ &
&  & & & \\
\hline
\end{tabular}
\end{table*}

\begin{rem}\label{rem:5.3}
\begin{enumerate}
 \item The unique additive $(12,2^{12},6)_4$-code is also known as \textit{dodecacode}. 
It is well known that the best Hermitian self-dual linear code is of parameters $[12,6,4]_4$.
 \item In~\cite{Var09}, four so-called \textit{additive circulant graph codes} 
of parameters $(30,2^{30},12)_4$ are constructed without classification. 
It is yet unknown if any of these four codes is inequivalent to the one listed in~\cite{GK04}.
\end{enumerate}
\end{rem}

\section{Construction from Self-Orthogonal Linear Codes}
\label{sec:HSelfOrtho}

It is well known (see~\cite[Th. 1.4.10]{HP03}) that a linear code 
$C$ having the parameters $[n,k,d]_4$ is Hermitian self-orthogonal 
if and only if the weights of its codewords are all even.

\begin{thm}\label{thm:6.1}
If $C$ is a Hermitian self-orthogonal code of parameters 
$[n,k,d]_4$, then there exists an asymmetric quantum code $Q$ 
with parameters $[[n,n-2k,d_z/d_x]]_4$, where
\begin{equation}\label{eq:4.2}
d_x = d_z = d(C^{\perp_{\H}})\text{.}
\end{equation}
\end{thm}

\begin{proof}
Seen as an additive code, $C$ is of parameters $(n,2^{2k},d)_4$ 
with $C^{\perp_{\tr}}$ being the code $C^{\perp_{\H}}$ seen as an 
$(n,2^{2n-2k},d^{\perp_{\tr}})$ additive code (see Remark~\ref{rem:4.1}). 
Applying Theorem~\ref{thm:3.5} by taking 
$C_1 = C^{\perp_{\tr}} = C_2$ to satisfy 
$C_1^{\perp_{\tr}} \subseteq C_2$ completes the proof.
\end{proof}

\begin{ex}\label{ex:4.3}
Let $n$ be an even positive integer. Consider the repetition code 
$[n,1,n]_4$ with weight enumerator $1+3Y^n$. Since the weights 
are all even, this $\F_{4}$-linear code is Hermitian self-orthogonal. 
We then have a quantum code $Q$ with parameters $\mathbf{[[n,n-2,2/2]]_{4}}$.
\end{ex}

Table~\ref{table:ClassSO} below presents the resulting asymmetric quantum 
codes based on the classification of self-orthogonal $\F_{4}$-linear 
codes of length up to 29 and of dimensions 3 up to 6 as presented 
in~\cite{BO05}. I. Bouyukliev~\cite{Bou09} shared with us the original data 
used in the said classification plus some additional results for 
lengths 30 and 31.

Given fixed length $n$ and dimension $k$, we only consider $[n,k,d]_4$-codes 
$C$ with maximal possible value for the minimum distances of their duals. 
For example, among 12 self-orthogonal $[10,4,4]_4$-codes, 
there are 4 distinct codes with $d^{\perp_{\H}}=3$ while the remaining 8 
codes have $d^{\perp_{\H}}=2$. We take only the first 4 codes.

The number of distinct codes that can be used for the construction of 
the asymmetric quantum codes for each set of given parameters is listed 
in the fourth column of the table.

\begin{table}
\caption{Asymmetric QECC from Classified Hermitian Self-Orthogonal $\F_{4}$-Linear Codes in~\cite{BO05}}
\label{table:ClassSO}
\centering
\begin{tabular}{| c | l | l | c | l |}
\hline
\textbf{No.} & \textbf{Code $C$} &\textbf{Code $Q$} & \textbf{num} \\
\hline
1 & $[6,3,4]_4$ & $[[6,0,3/3]]_4$ & $1$ \\ % inserting body of the table
2 & $[7,3,4]_4$ & $[[7,1,3/3]]_4$ & $1$ \\
3 & $[8,3,4]_4$ & $[[8,2,3/3]]_4$ & $1$ \\
4 & $[8,4,4]_4$ & $[[8,0,4/4]]_4$ & $1$ \\
5 & $[9,3,6]_4$ & $[[9,3,3/3]]_4$ & $1$ \\

6 & $[9,4,4]_4$ & $[[9,1,3/3]]_4$ & $2$ \\
7 & $[10,3,6]_4$ & $[[10,4,3/3]]_4$ & $1$ \\
8 & $[10,4,4]_4$ & $[[10,2,3/3]]_4$ & $3$ \\
9 & $[10,5,4]_4$ & $[[10,0,4/4]]_4$ & $2$ \\
10 & $[11,3,6]_4$ & $[[11,5,3/3]]_4$ & $1$ \\

11 & $[11,4,4]_4$ & $[[11,3,3/3]]_4$ & $3$ \\
12 & $[11,4,6]_4$ & $[[11,3,3/3]]_4$ & $1$ \\
13 & $[11,5,4]_4$ & $[[11,1,3/3]]_4$ & $6$ \\
14 & $[12,3,8]_4$ & $[[12,6,3/3]]_4$ & $1$ \\
15 & $[12,4,6]_4$ & $[[12,4,4/4]]_4$ & $1$ \\

16 & $[12,5,6]_4$ & $[[12,2,4/4]]_4$ & $1$ \\
17 & $[12,6,4]_4$ & $[[12,0,4/4]]_4$ & $5$ \\
18 & $[13,3,8]_4$ & $[[13,7,3/3]]_4$ & $1$ \\
19 & $[13,4,8]_4$ & $[[13,5,3/3]]_4$ & $5$ \\
20 & $[13,5,6]_4$ & $[[13,3,4/4]]_4$ & $1$ \\

21 & $[13,6,6]_4$ & $[[13,1,5/5]]_4$ & $1$ \\
22 & $[14,3,10]_4$ & $[[14,8,3/3]]_4$ & $1$ \\
23 & $[14,4,8]_4$ & $[[14,6,4/4]]_4$ & $1$ \\
24 & $[14,5,8]_4$ & $[[14,4,4/4]]_4$ & $4$ \\
25 & $[14,6,6]_4$ & $[[14,2,5/5]]_4$ & $1$ \\

26 & $[14,7,6]_4$ & $[[14,0,6/6]]_4$ & $1$ \\
27 & $[15,3,10]_4$ & $[[15,9,3/3]]_4$ & $1$ \\
28 & $[15,4,8]_4$ & $[[15,7,3/3]]_4$ & $189$ \\
29 & $[15,5,8]_4$ & $[[15,5,4/4]]_4$ & $26$ \\
30 & $[15,6,8]_4$ & $[[15,3,5/5]]_4$ & $3$ \\

31 & $[16,3,12]_4$ & $[[16,10,3/3]]_4$ & $1$ \\
32 & $[16,4,10]_4$ & $[[16,8,3/3]]_4$ & $38$ \\
33 & $[16,5,8]_4$ & $[[16,6,4/4]]_4$ & $519$ \\
34 & $[16,6,8]_4$ & $[[16,4,4/4]]_4$ & $697$ \\
35 & $[17,3,12]_4$ & $[[17,11,2/2]]_4$ & $4$ \\

36 & $[17,4,12]_4$ & $[[17,9,4/4]]_4$ & $1$ \\
37 & $[17,5,10]_4$ & $[[17,7,4/4]]_4$ & $27$ \\
38 & $[18,3,12]_4$ & $[[18,12,2/2]]_4$ & $45$ \\
39 & $[18,4,12]_4$ & $[[18,10,3/3]]_4$ & $11$ \\
40 & $[18,6,10]_4$ & $[[18,6,5/5]]_4$ & $2$ \\

41 & $[19,3,12]_4$ & $[[19,13,2/2]]_4$ & $185$ \\
42 & $[19,4,12]_4$ & $[[19,11,3/3]]_4$ & $2570$ \\
43 & $[20,3,14]_4$ & $[[20,14,2/2]]_4$ & $10$ \\
44 & $[20,5,12]_4$ & $[[20,10,3/3]]_4$ & $4$ \\
45 & $[21,3,16]_4$ & $[[21,15,3/3]]_4$ & $1$ \\

46 & $[21,4,14]_4$ & $[[21,13,3/3]]_4$ & $212$ \\
47 & $[21,5,12]_4$ & $[[21,11,3/3]]_4$ & $3$ \\
48 & $[22,3,16]_4$ & $[[22,16,3/3]]_4$ & $4$ \\
49 & $[22,5,14]_4$ & $[[22,12,4/4]]_4$ & $67$ \\
50 & $[23,3,16]_4$ & $[[23,17,2/2]]_4$ & $46$ \\

51 & $[23,4,16]_4$ & $[[23,15,3/3]]_4$ & $1$ \\
52 & $[24,3,16]_4$ & $[[24,18,2/2]]_4$ & $614$ \\
53 & $[24,4,16]_4$ & $[[24,16,3/3]]_4$ & $20456$ \\
54 & $[25,3,18]_4$ & $[[25,19,2/2]]_4$ & $6$ \\
55 & $[25,4,16]_4$ & $[[25,17,3/3]]_4$ & $19$ \\

56 & $[26,3,18]_4$ & $[[26,20,2/2]]_4$ & $185$ \\
57 & $[26,4,18]_4$ & $[[26,18,3/3]]_4$ & $14$ \\
58 & $[27,3,20]_4$ & $[[27,21,2/2]]_4$ & $2$ \\
59 & $[28,3,20]_4$ & $[[28,22,2/2]]_4$ & $46$ \\
60 & $[28,4,20]_4$ & $[[28,20,3/3]]_4$ & $1$ \\

61 & $[29,3,20]_4$ & $[[29,23,2/2]]_4$ & $850$ \\
62 & $[29,4,20]_4$ & $[[29,21,3/3]]_4$ & $11365$ \\
63 & $[30,5,20]_4$ & $[[30,20,3/3]]_4$ & $\geq90$ \\
64 & $[31,4,22]_4$ & $[[31,23,3/3]]_4$ & $1$ \\
\hline
\end{tabular}
\end{table}

Comparing some entries in Table~\ref{table:ClassSO}, say, numbers 5 and 6, 
we notice that the $[[9,3,3/3]]_4$-code has better parameters than 
the $[[9,1,3/3]]_4$-code does. Both codes are included in the table 
in the interest of preserving the information on precisely how many of 
such codes there are from the classification result.

In~\cite[Table 7]{GG09}, examples of $\F_{4}$-linear self-dual codes for 
even lengths $2 \leq n \leq 80$ are presented. Table~\ref{table:ClassSD} lists 
down the resulting asymmetric quantum codes for $32 \leq n \leq 80$.
\begin{table}
\caption{Asymmetric QECC from Hermitian Self-Dual $\F_{4}$-Linear Codes 
based on~\cite[Table 7]{GG09} for $32 \leq n \leq 80$}
\label{table:ClassSD}
\centering
\begin{tabular}{| l | l || l | l |}
\hline
\textbf{Code $C$} &\textbf{Code $Q$} & \textbf{Code $C$} &\textbf{Code $Q$}\\
\hline
$[32,16,10]_4$ & $[[32,0,10/10]]_4$ & $[58,29,14]_4$ & $[[58,0,14/14]]_4$\\
$[34,17,10]_4$ & $[[34,0,10/10]]_4$ & $[60,30,16]_4$ & $[[60,0,16/16]]_4$\\
$[36,18,12]_4$ & $[[36,0,12/12]]_4$ & $[62,31,18]_4$ & $[[62,0,18/18]]_4$\\
$[38,19,12]_4$ & $[[38,0,12/12]]_4$ & $[64,32,16]_4$ & $[[64,0,16/16]]_4$\\
$[40,20,12]_4$ & $[[40,0,12/12]]_4$ & $[66,33,16]_4$ & $[[66,0,16/16]]_4$\\

$[42,21,12]_4$ & $[[42,0,12/12]]_4$ & $[68,34,18]_4$ & $[[68,0,18/18]]_4$\\
$[44,22,12]_4$ & $[[44,0,12/12]]_4$ & $[70,35,18]_4$ & $[[70,0,18/18]]_4$\\
$[46,23,14]_4$ & $[[46,0,14/14]]_4$ & $[72,36,18]_4$ & $[[72,0,18/18]]_4$\\
$[48,24,14]_4$ & $[[48,0,14/14]]_4$ & $[74,37,18]_4$ & $[[74,0,18/18]]_4$\\
$[50,25,14]_4$ & $[[50,0,14/14]]_4$ & $[76,38,18]_4$ & $[[76,0,18/18]]_4$\\

$[52,26,14]_4$ & $[[52,0,14/14]]_4$ & $[78,39,18]_4$ & $[[78,0,18/18]]_4$\\
$[54,27,16]_4$ & $[[54,0,16/16]]_4$ & $[80,40,20]_4$ & $[[80,0,20/20]]_4$\\
$[56,28,14]_4$ & $[[56,0,14/14]]_4$ & & \\
\hline
\end{tabular}
\end{table}

For parameters other than those listed in Table~\ref{table:ClassSO}, 
we do not have complete classification just yet. The Q-extension 
program described in~\cite{Bou07} can be used to extend the 
classification effort given sufficient resources. 
Some classifications based on the optimality of the minimum 
distances of the codes can be found in~\cite{BGV04} and 
in~\cite{ZLL09}, although when used in the construction of 
asymmetric quantum codes using our framework, they do not yield 
good $d_{z} = d_{x}$ relative to the length $n$.

Many other $\F_{4}$-linear self-orthogonal codes are known. 
Examples can be found in~\cite[Table II]{CRSS98}, 
\cite{MLZF07}, as well as from the list of best known linear codes 
(BKLC) over $\F_4$ as explained in~\cite{Gr09}.

Table~\ref{table:RandomSO} presents more 
examples of asymmetric quantum codes constructed based on known 
self-orthogonal linear codes up to length $n=40$. The list of codes in 
Table~\ref{table:RandomSO} is by no means exhaustive. It may be 
possible to find asymmetric codes with better parameters.

\begin{table}
\caption{Asymmetric QECC from Hermitian Self-orthogonal $\F_{4}$-Linear Codes for $n \leq 40$}
\label{table:RandomSO}
\centering
\begin{tabular}{| c | l | l | l |}
\hline
\textbf{No.} & \textbf{Code $C$} & \textbf{Code $Q$} & \textbf{Ref.} \\
\hline
1 & $[5,2,4]_4$ & $\mathbf{[[5,1,3/3]]_4}$ & \cite[BKLC]{Gr09}\\
2 & $[6,2,2]_4$ & $[[6,2,2/2]]_4$ & \cite{Bou09}\\
3 & $[8,2,6]_4$ & $[[8,4,2/2]]_4$ & \cite[BKLC]{Gr09}\\
4 & $[10,2,8]_4$ & $[[10,6,2/2]]_4$ & \cite[BKLC]{Gr09}\\
5 & $[14,7,6]_4$ & $[[14,0,6/6]]_4$  & \cite[Table 7]{GG09}\\

6 & $[15,2,12]_4$ & $[[15,11,2/2]]_4$ & \cite[BKLC]{Gr09}\\
7 & $[16,8,6]_4$ & $[[16,0,6/6]]_4$ & \cite[Table 7]{GG09}\\
8 & $[20,2,16]_4$ & $[[20,16,2/2]]_4$ & \cite[BKLC]{Gr09}\\
9 & $[20,5,12]_4$ & $[[20,10,4/4]]_4$ & \cite[Table II]{CRSS98}\\
10 & $[20,9,8]_4$ & $[[20,2,6/6]]_4$  & \cite[p.~788]{MLZF07}\\

11 & $[20,10,8]_4$ & $[[20,0,8/8]]_4$  & \cite[Table 7]{GG09}\\
12 & $[22,8,10]_4$ & $[[22,6,5/5]]_4$ & \cite{LM09}\\
13 & $[22,10,8]_4$ & $[[22,2,6/6]]_4$ & \cite{LM09}\\
14 & $[22,11,8]_4$ & $[[22,0,8/8]]_4$  & \cite[Table 7]{GG09}\\
15 & $[23,8,10]_4$ & $[[23,7,5/5]]_4$ & \cite{LM09}\\

16 & $[23,8,12]_4$ & $[[23,7,5/5]]_4$ & \cite[BKLC]{Gr09}\\ 
17 & $[23,10,8]_4$ & $[[23,3,6/6]]_4$ & \cite{LM09}\\
18 & $[24,5,16]_4$ & $[[24,14,3/3]]_4$ & \cite[BKLC]{Gr09}\\
19 & $[24,8,10]_4$ & $[[24,8,5/5]]_4$ & \cite{LM09}\\
20 & $[24,9,12]_4$ & $[[24,6,6/6]]_4$ & \cite[BKLC]{Gr09}\\

21 & $[24,12,8]_4$ & $[[24,0,8/8]]_4$  & \cite[Table 7]{GG09}\\
22 & $[25,2,20]_4$ & $[[25,21,2/2]]_4$ & \cite[BKLC]{Gr09}\\
23 & $[25,5,16]_4$ & $[[25,15,4/4]]_4$ & \cite[Table II]{CRSS98}\\
24 & $[25,8,10]_4$ & $[[25,9,5/5]]_4$ & \cite{LM09}\\
25 & $[25,10,12]_4$ & $[[25,5,7/7]]_4$ & \cite{LM09}\\

26 & $[26,2,20]_4$ & $[[26,22,2/2]]_4$ & \cite[BKLC]{Gr09}\\
27 & $[26,6,16]_4$ & $[[26,14,4/4]]_4$ & \cite[BKLC]{Gr09}\\
28 & $[26,9,10]_4$ & $[[26,8,5/5]]_4$ & \cite{LM09}\\
29 & $[26,10,10]_4$ & $[[26,6,6/6]]_4$ & \cite{LM09}\\
30 & $[26,11,12]_4$ & $[[26,4,8/8]]_4$  & \cite[BKLC]{Gr09}\\

31 & $[27,6,16]_4$ & $[[27,15,3/3]]_4$ & \cite[BKLC]{Gr09}\\
32 & $[27,9,10]_4$ & $[[27,9,5/5]]_4$ & \cite{LM09}\\
33 & $[27,10,10]_4$ & $[[27,7,6/6]]_4$ & \cite{LM09}\\
34 & $[27,12,12]_4$ & $[[27,3,9/9]]_4$ & \cite[BKLC]{Gr09}\\
35 & $[28,7,16]_4$ & $[[28,14,5/5]]_4$ & \cite[Table II]{CRSS98}\\

36 & $[28,8,12]_4$ & $[[28,12,5/5]]_4$ & \cite{LM09}\\
37 & $[28,10,10]_4$ & $[[28,8,6/6]]_4$ & \cite{LM09}\\
38 & $[28,13,12]_4$ & $[[28,2,10/10]]_4$ & \cite[BKLC]{Gr09}\\
39 & $[29,8,16]_4$ & $[[29,13,5/5]]_4$ & \cite[BKLC]{Gr09}\\
40 & $[29,11,10]_4$ & $[[29,7,6/6]]_4$ & \cite{LM09}\\

41 & $[29,14,12]_4$ & $[[29,1,11/11]]_4$ & \cite[BKLC]{Gr09}\\
42 & $[30,2,24]_4$ & $[[30,26,2/2]]_4$ & \cite[BKLC]{Gr09}\\
43 & $[30,5,20]_4$ & $[[30,20,4/4]]_4$ & \cite[Table 13.2]{NRS06}\\
44 & $[30,9,12]_4$ & $[[30,12,5/5]]_4$ & \cite{LM09}\\
45 & $[30,11,10]_4$ & $[[30,8,6/6]]_4$ & \cite{LM09}\\

46 & $[30,15,12]_4$ & $[[30,0,12/12]]_4$ & \cite[Table 7]{GG09}\\
47 & $[31,9,16]_4$ & $[[31,13,6/6]]_4$ & \cite[BKLC]{Gr09}\\
48 & $[32,6,20]_4$ & $[[32,20,4/4]]_4$ & \cite[Table 3]{GO98}\\
49 & $[32,9,14]_4$ & $[[32,14,5/5]]_4$ & \cite{LM09}\\
50 & $[32,11,12]_4$ & $[[32,10,6/6]]_4$ & \cite{LM09}\\

51 & $[33,7,20]_4$ & $[[33,19,4/4]]_4$ & \cite[BKLC]{Gr09}\\
52 & $[33,9,14]_4$ & $[[33,15,5/5]]_4$ & \cite{LM09}\\
53 & $[33,10,16]_4$ & $[[33,13,6/6]]_4$ & \cite[BKLC]{Gr09}\\
54 & $[33,12,14]_4$ & $[[33,9,7/7]]_4$ & \cite[BKLC]{Gr09}\\
55 & $[33,15,12]_4$ & $[[33,3,9/9]]_4$ & \cite[BKLC]{Gr09}\\

56 & $[34,9,18]_4$ & $[[34,16,6/6]]_4$ & \cite{LM09}\\
57 & $[34,13,14]_4$ & $[[34,8,8/8]]_4$ & \cite[BKLC]{Gr09}\\
58 & $[34,16,12]_4$ & $[[34,2,10/10]]_4$ & \cite[BKLC]{Gr09}\\
59 & $[35,5,24]_4$ & $[[35,25,3/3]]_4$ & \cite[BKLC]{Gr09}\\
60 & $[35,8,20]_4$ & $[[35,19,5/5]]_4$ & \cite[BKLC]{Gr09}\\

61 & $[35,11,14]_4$ & $[[35,13,6/6]]_4$ & \cite{LM09}\\
62 & $[35,17,12]_4$ & $[[35,1,11/11]]_4$ & \cite[BKLC]{Gr09}\\
63 & $[36,9,16]_4$ & $[[36,18,4/4]]_4$ & \cite{LM09}\\
64 & $[36,11,14]_4$ & $[[36,14,6/6]]_4$ & \cite{LM09}\\
65 & $[37,9,20]_4$ & $[[37,19,5/5]]_4$ & \cite[BKLC]{Gr09}\\

66 & $[37,18,12]_4$ & $[[37,1,11/11]]_4$ & \cite[BKLC]{Gr09}\\
67 & $[38,6,24]_4$ & $[[38,26,3/3]]_4$ & \cite[BKLC]{Gr09}\\
68 & $[38,11,18]_4$ & $[[38,16,6/6]]_4$ & \cite[BKLC]{Gr09}\\
69 & $[39,12,18]_4$ & $[[39,15,7/7]]_4$ & \cite[BKLC]{Gr09}\\
70 & $[39,4,28]_4$ & $[[39,31,3/3]]_4$ & \cite[BKLC]{Gr09}\\

71 & $[39,7,24]_4$ & $[[39,25,4/4]]_4$ & \cite[BKLC]{Gr09}\\
72 & $[40,5,28]_4$ & $[[40,30,4/4]]_4$ & \cite[Table II]{CRSS98}\\
73 & $[40,15,16]_4$ & $[[40,10,7/7]]_4$ & \cite[BKLC]{Gr09}\\
\hline 
\end{tabular}
\end{table}

For lengths larger than $n=40$,~\cite{GO98} provides some known 
$\F_{4}$-linear codes of dimension 6 that belong to the class of 
\textit{quasi-twisted codes}. 
Based on the weight distribution of these codes~\cite[Table 3]{GO98}, 
we know which ones of them are self-orthogonal. Applying 
Theorem~\ref{thm:3.5} to them yields the 12 quantum codes listed in 
Table~\ref{table:ClassQT}.

\begin{table}
\caption{Asymmetric QECC from Hermitian Self-orthogonal Quasi-Twisted Codes Found in~\cite{GO98}}
\label{table:ClassQT}
\centering
\begin{tabular}{| l | l || l | l |}
\hline
\textbf{Code $C$} & \textbf{Code $Q$} & 
\textbf{Code $C$} & \textbf{Code $Q$}\\
\hline
$[48,6,32]_4$   & $[[48,36,3/3]]_4$   &  $[144,6,104]_4$ & $[[144,132,3/3]]_4$\\
$[78,6,56]_4$   & $[[78,66,4/4]]_4$   &  $[150,6,108]_4$ & $[[150,138,3/3]]_4$\\
$[102,6,72]_4$  & $[[102,90,3/3]]_4$  &  $[160,6,116]_4$ & $[[160,148,3/3]]_4$\\
$[112,6,80]_4$  & $[[112,100,3/3]]_4$ &  $[182,6,132]_4$ & $[[182,170,3/3]]_4$\\
$[120,6,86]_4$  & $[[120,108,4/4]]_4$ &  $[192,6,138]_4$ & $[[192,180,3/3]]_4$\\
$[132,6,94]_4$  & $[[132,120,3/3]]_4$ &  $[200,6,144]_4$ & $[[200,188,3/3]]_4$\\
\hline 
\end{tabular}
\end{table}

Another family of codes that we can use is the 
\textit{MacDonald codes}, commonly denoted by $C_{k,u}$ 
with $k>u>0$. The MacDonald codes are linear codes with parameters 
$[(q^{k}-q^{u})/(q-1),k,q^{k-1}-q^{u-1}]_{q}$. Some 
historical background and a construction of their generator 
matrices can be found in~\cite{BD03}. It is known that these codes 
are \textit{two-weight codes}. That is, they have nonzero 
codewords of only two possible weights. In~\cite[Figures 1a and 2a]{CaKa86}, 
the MacDonald codes are labeled SU1. There are $q^{k}-q^{k-u}$ 
codewords of weight $q^{k-1}-q^{u-1}$ and $q^{k-u}-1$ 
codewords of weight $q^{k-1}$.

The MacDonald codes satisfy the equality of the 
\textit{Griesmer bound} which says that, for any $[n,k\geq1,d]_q$-code,
\begin{equation}\label{eq:4.3}
 n \geq \sum_{i=0}^{k-1}\left \lceil \frac{d}{q^{i}} \right \rceil \text{.}
\end{equation}

\begin{ex}\label{ex:4.4}
For $q=4,k>u>1$, the MacDonald codes are self-orthogonal 
since both $4^{k-1}-4^{u-1}$ and 
$4^{k-1}$ are even. For a 
$[(4^{k}-4^{u})/3,k,4^{k-1}-4^{u-1}]_{4}$-code 
$C_{k,u}$, we know (see~\cite[Lemma 4]{BD03}) that 
$d(C_{k,u}^{\perp_{\H}}) \geq 3$. Using $C_{k,u}$ and 
applying Theorem~\ref{thm:6.1}, we get an asymmetric quantum code $Q$ 
with parameters
\begin{equation*}
[[(4^{k}-4^{u})/3,((4^{k}-4^{u})/3)-2k,(\geq 3)/(\geq 3)]]_4 \text{.}
\end{equation*}

For $k\leq 5$ we have the following more explicit examples. 
The weight enumerator is written in an abbreviated form. 
For instance, $(0,1),(12,60),(16,3)$ means that the corresponding 
code has 1 codeword of weight 0, 60 codewords of weight 12 and 
3 codewords of weight 16.

\begin{enumerate}
\item For $k=3,u=2$, we have the $[16,3,12]_4$-code with 
weight enumerator $(0,1),(12,60),(16,3)$. The resulting 
asymmetric QECC is a $[[16,10,3/3]]_4$-code. This code 
is listed as number 31 in Table~\ref{table:ClassSO}.
\item For $k=4,u=3$, we have the $[64,4,48]_4$-code with 
weight enumerator $(0,1),(48,252),(64,3)$. The resulting 
asymmetric QECC is a $[[64,56,3/3]]_4$-code.
\item For $k=4,u=2$, we have the $[80,4,60]_4$-code with 
weight enumerator $(0,1),(60,240),(64,15)$. The resulting 
asymmetric QECC is a $[[80,72,3/3]]_4$-code.
\item For $k=5,u=4$, we have the $[256,5,192]_4$-code with 
weight enumerator $(0,1),(192,1020),(256,3)$. The resulting 
asymmetric QECC is a $[[256,246,3/3]]_4$-code.
\item For $k=5,u=3$, we have the $[320,5,240]_4$-code with 
weight enumerator $(0,1),(240,1008),(256,15)$. The resulting 
asymmetric QECC is a $[[320,310,3/3]]_4$-code.
\item For $k=5,u=2$, we have the $[336,5,252]_4$-code with 
weight enumerator $(0,1),(252,960),(256,63)$. The resulting 
asymmetric QECC is a $[[336,326,3/3]]_4$-code.
\end{enumerate}
\end{ex}

\section{Construction from Nested Linear Cyclic Codes}\label{sec:Cyclic}
The asymmetric quantum codes that we have constructed so far have 
$d_{z} = d_{x}$. From this section onward, we construct asymmetric 
quantum codes with $d_{z} \geq d_{x}$. In most cases, $d_{z} > d_{x}$.

It is well established that, under the natural correspondence 
of vectors and polynomials, the study of cyclic codes in 
$\F_q^n$ is equivalent to the study of ideals in the residue class ring
\begin{equation*}
 \mathcal{R}_n = \F_q[x]/(x^n-1) \text{.}
\end{equation*}

The study of ideals in $\mathcal{R}_n$ depends on factoring $x^n-1$. 
Basic results concerning and the properties of cyclic codes can be 
found in~\cite[Ch. 4]{HP03} or~\cite[Ch. 7]{MS77}. A cyclic code $C$ 
is a subset of a cyclic code $D$ of equal length over $\F_q$ 
if and only if the generator polynomial of $D$ divides the generator 
polynomial of $C$. Both polynomials divide $x^n-1$. Once the 
factorization of $x^n-1$ into irreducible polynomials is known, 
the nestedness property becomes apparent.

We further require that $n$ be relatively prime to $4$ 
to exclude the so-called repeated-root cases since the resulting 
cyclic codes when $n$ is not relatively prime to $4$ have 
inferior parameters. See~\cite[p. 976]{Cha98} for comments 
and references regarding this matter.

\begin{thm}\label{thm:7.1}
Let $C$ and $D$ be cyclic codes of parameters $[n,k_1,d_1]_4$ 
and $[n,k_2,d_2]_4$, respectively, with $C \subseteq D$, 
then there exists an asymmetric quantum code $Q$ 
with parameters $[[n,k_2-k_1,d_z/d_x]]_4$, where
\begin{equation}\label{eq:7.1}
\left\lbrace d_z,d_x \right\rbrace = 
\left\lbrace d(C^{\perp_{\H}}),d_2 \right\rbrace \text{.}
\end{equation}
\end{thm}

\begin{proof}
Apply Theorem~\ref{thm:3.5} by taking $C_1 = C^{\perp_{\tr}}$ 
and $C_2 = D$. Since $C$ is an $[n,k_1,d_1]_4$ code, $C$ 
is an additive code of parameters $(n,2^{2k_1},d_1)_4$. 
Similarly, $D$ is an additive code of parameters $(n,2^{2k_2},d_2)_4$. 
The values for $d_z$ and $d_x$ can be verified by simple calculations.
\end{proof}

\begin{ex}\label{ex:7.2}
Let $C$ be the repetition $[n,1,n]_4$-code generated by 
the polynomial $(x^{n}+1)/(x+1)$. If we take $C=D$ in Theorem~\ref{thm:7.1}, 
then we get a quantum code $Q$ with parameters $\mathbf{[[n,0,n/2]]_4}$.
\end{ex}

Tables~\ref{table:Cyclic} and~\ref{table:Cyclic2} list examples 
of asymmetric quantum codes constructed from nested cyclic codes 
up to $n=25$. We exclude the case $C=D$ since the parameters of the 
resulting quantum code $Q$ are $[[n,0,d(C)/2]]_{4}$ which are 
never better than those of the code $Q$ in Example~\ref{ex:7.2}. 
Among the resulting codes $Q$ of equal length and dimension, 
we choose one with the largest $d_z,d_x$ values. For codes 
$Q$ with equal length and distances, we choose one 
with the largest dimension.

%For lengths $n=15,21$, there are $2^9=512$ cyclic codes each. 
%The comparison of parameters was done heuristically, not exhaustively.
%
\begin{table*}
\caption{Asymmetric QECC from Nested Cyclic Codes}
\label{table:Cyclic}
\centering
\begin{tabular}{| c | l | l | l |}
\hline
\textbf{No.} & \textbf{Codes $C$ and $D$} & \textbf{Generator Polynomials of $C$ and of $D$} & \textbf{Code $Q$} \\
\hline
1 & $[3,1,3]_4$ & $(x+1)(x+\omega)$ & $\mathbf{[[3,1,2/2]]_4}$ \\ % inserting body of the table
  & $[3,2,2]_4$ & $(x+1)$ & \\
2 & $[5,1,5]_4$ & $(x^2+\omega x+1)(x^2+\omega^2 x+1)$ & $\mathbf{[[5,2,3/2]]_{4}}$ \\
  & $[5,3,3]_4$ & $(x^2+\omega^2 x+1)$ & \\
3 & $[7,1,7]_4$ & $(x^3+x+1)(x^3+x^2+1)$ & $[[7,3,3/2]]_4$ \\ 
  & $[7,4,3]_4$ & $(x^3+x+1)$ & \\
4 & $[7,3,4]_4$ & $(x^3+x+1)(x+1)$ & $[[7,1,3/3]]_4$ \\ 
  & $[7,4,3]_4$ & $(x^3+x+1)$ & \\

5 & $[9,1,9]_4$ & $(x+\omega)(x+\omega^2)(x^3+\omega)(x^3+\omega^2)$ & $[[9,1,6/2]]_4$ \\ 
  & $[9,2,6]_4$ & $(x+\omega^2)(x^3+\omega)(x^3+\omega^2)$ & \\
6 & $[9,1,9]_4$ & $(x+\omega)(x+\omega^2)(x^3+\omega)(x^3+\omega^2)$ & $[[9,4,3/2]]_4$ \\ 
  & $[9,5,3]_4$ & $(x+\omega)(x^3+\omega)$ & \\
7 & $[9,1,9]_4$ & $(x+\omega)(x+\omega^2)(x^3+\omega)(x^3+\omega^2)$ & $\mathbf{[[9,7,2/2]]_4}$ \\ 
  & $[9,8,2]_4$ & $(x+\omega)$ & \\

8 & $[11,5,6]_4$ & $(x+1)(x^5+\omega^2 x^4+x^3+x^2+\omega x+1)$ & $[[11,1,5/5]]_4$ \\ 
  & $[11,6,5]_4$ & $(x^5+\omega^2 x^4+x^3+x^2+\omega x+1)$ & \\
9 & $[11,1,11]_4$ & $(x^{11}+1)/(x+1)$ & $[[11,5,5/2]]_4$ \\ 
  & $[11,6,5]_4$ & $(x^5+\omega x^4+x^3+x^2+\omega^2 x+1)$ & \\

10 & $[13,6,6]_4$ & $(x+1)(x^6 +\omega x^5 +\omega^2 x^3 +\omega x +1)$ & $[[13,1,5/5]]_4$ \\ 
  & $[13,7,5]_4$ & $(x^6 +\omega x^5 +\omega^2 x^3 +\omega x +1)$ & \\
11 & $[13,1,13]_4$ & $(x^{13}+1)/(x+1)$ & $[[13,6,5/2]]_4$ \\ 
  & $[13,7,5]_4$ & $(x^6 +\omega x^5 +\omega^2 x^3 +\omega x +1)$ & \\

12 & $[15,3,11]_4$ & $(x^{15}+1)/((x+1)(x^2 + \omega^2 x + \omega^2))$ & $[[15,1,9/3]]_4$ \\ 
  & $[15,4,9]_4$ & $(x^{15}+1)/((x+1)(x+\omega)(x^2 + \omega^2 x + \omega^2))$ & \\
13 & $[15,6,8]_4$ & $(x^9 +\omega x^8 +x^7+x^5 +\omega x^4 + \omega^2 x^2 +\omega^2x + 1)$ & $[[15,1,7/5]]_4$ \\ 
  & $[15,7,7]_4$ & $(x^8 +\omega^2 x^7 +\omega x^6 +\omega x^5 +\omega^2 x^4+x^3+x^2+\omega x + 1)$ & \\
14 & $[15,7,7]_4$ & $(x^8 +\omega^2 x^7 +\omega x^6 +\omega x^5 +\omega^2 x^4+x^3+x^2+\omega x + 1)$ & $[[15,1,6/6]]_4$ \\ 
   & $[15,8,6]_4$ & $(x^7 + x^6+\omega x^4+x^2+\omega^2 x +\omega^2)$ & \\
15 & $[15,1,15]_4$ & $(x^{15}+1)/(x+1)$ & $[[15,2,11/2]]_4$ \\ 
  & $[15,3,11]_4$ & $(x^{15}+1)/((x+1)(x^2+\omega^2 x+\omega^2))$ & \\

16 & $[15,3,11]_4$ & $(x^{15}+1)/((x+1)(x^2+\omega^2 x+\omega^2))$ & $[[15,2,8/3]]_4$ \\ 
  & $[15,5,8]_4$ & $(x^{15}+1)/((x+1)(x^2+\omega^2 x+\omega^2)(x^2+\omega^2 x+1))$ & \\
17 & $[15,6,8]_4$ & $(x^{15}+1)/((x^2+\omega x+\omega)(x^2+\omega^2 x+\omega^2)(x^2+\omega^2 x+1))$ & $[[15,2,6/5]]_4$ \\ 
  & $[15,8,6]_4$ & $(x^7 + x^6+\omega x^4+x^2+\omega^2 x +\omega^2)$ & \\
18 & $[15,1,15]_4$ & $(x^{15}+1)/(x+1)$ & $[[15,3,9/2]]_4$ \\ 
  & $[15,4,9]_4$ & $(x^{15}+1)/((x+1)(x+\omega)(x^2+\omega^2 x+\omega^2))$ & \\
19 &$[15,8,6]_4$ &  $x^{7} + \omega x^{6} + \omega^2 x^{4} + \omega^2 x^{2} + \omega x + \omega^2$ & $[[15,4,7/3]]_4$ \\ 
   &$[15,12,3]_4$& $x^{3} + x^{2} + \omega^{2}$ & \\
20 & $[15,1,15]_4$ & $(x^{15}+1)/(x+1)$ & $[[15,4,8/2]]_4$ \\ 
  & $[15,5,8]_4$ & $(x^{15}+1)/((x+1)(x^2+\omega^2 x+\omega^2)(x^2+\omega^2 x+1))$ & \\

21 &$[15,4,10]_4$& $(x^{15}+1)/(x^4 + \omega^2 x^3 + \omega x^2 + \omega x + w)$ & $[[15,5,5/4]]_4$ \\ 
  &$[15,9,5]_4$&  $x^6 + \omega^2 x^5 + \omega^2 x^4 + x^3 + x^2 + \omega x + 1$ & \\
22 & $[15,1,15]_4$ & $(x^{15}+1)/(x+1)$ & $[[15,6,7/2]]_4$ \\ 
  & $[15,7,7]_4$ & $(x^8 +\omega^2 x^7 +\omega x^6 +\omega x^5 +\omega^2 x^4+x^3+x^2+\omega x + 1)$ & \\
23 & $[15,3,11]_4$ & $(x^{15}+1)/((x+1)(x^2 + \omega^2 x + \omega^2))$ & $[[15,6,5/3]]_4$ \\ 
  & $[15,9,5]_4$ & $(x^6 + \omega x^5 + x^4 + x^3 + \omega^2 x^2 + \omega^2 x + 1)$ & \\
24 & $[15,1,15]_4$ & $(x^{15}+1)/(x+1)$ & $[[15,7,6/2]]_4$ \\ 
  & $[15,8,6]_4$ & $(x^7 + x^6+\omega x^4+x^2+\omega^2 x +\omega^2)$ & \\
25 & $[15,1,15]_4$ & $(x^{15}+1)/(x+1)$ & $[[15,8,5/2]]_4$ \\ 
  & $[15,9,5]_4$ & $(x^6 + \omega x^5 + x^4 + x^3 + \omega^2 x^2 + \omega^2 x + 1)$ & \\

26 &$[15,3,11]_4$& $(x^{15}+1)/(x^3 + \omega^2 x^2 + \omega^2)$ & $[[15,9,3/3]]_4$ \\ 
  &$[15,12,3]_4$& $x^3 + x^2 + \omega^2$ & \\
27 &$[15,1,15]_4$& $(x^{15}+1)/(x+1)$ & $[[15,11,3/2]]_4$ \\ 
   &$[15,12,3]_4$& $x^3 + x^2 + \omega^2$ & \\
28 & $[15,1,15]_4$ & $(x^{15}+1)/(x+1)$ & $\mathbf{[[15,13,2/2]]_4}$ \\ 
  & $[15,14,2]_4$ & $(x+ \omega)$ & \\
29 & $[17,12,4]_4$ & $(x^5 + \omega x^3 + \omega x^2 + 1)$ & $[[17,1,9/4]]_4$ \\ 
  & $[17,13,4]_4$ & $(x^4 + x^3 + \omega^2 x^2 + x + 1)$ & \\
30 & $[17,8,8]_4$ & $(x^9 +\omega x^8 +\omega^2 x^7 +\omega^2 x^6 +\omega^2 x^3 +\omega^2 x^2 +\omega x + 1)$ & $[[17,1,7/7]]_4$ \\ 
  & $[17,9,7]_4$ & $(x^8 +\omega^2 x^7 +\omega^2 x^5 +\omega^2 x^4 +\omega^2 x^3 +\omega^2 x + 1)$ & \\

31 & $[17,1,17]_4$ & $(x^{17}+1)/(x+1)$ & $[[17,4,9/2]]_4$ \\ 
  & $[17,5,9]_4$ & $(x^{17}+1)/(x^5 +\omega^2 x^4 +\omega^2 x^3 +\omega^2 x^2 +\omega^2 x + 1)$ & \\
32 & $[17,4,12]_4$ & $(x^{17}+1)/(x^4+x^3+\omega x^2+x+1)$ & $[[17,4,8/4]]_4$ \\ 
  & $[17,8,8]_4$ & $(x^9 +\omega x^8 +\omega^2 x^7 +\omega^2 x^6 +\omega^2 x^3 +\omega^2 x^2 +\omega x + 1)$ & \\
33 & $[17,4,12]_4$ & $(x^{17}+1)/(x^4+x^3+\omega x^2+x+1)$ & $[[17,5,7/4]]_4$ \\ 
  & $[17,9,7]_4$ & $(x^8 +\omega^2 x^7 +\omega^2 x^5 +\omega^2 x^4 +\omega^2 x^3 +\omega^2 x + 1)$ & \\
34 & $[17,1,17]_4$ & $(x^{17}+1)/(x+1)$ & $[[17,8,7/2]]_4$ \\ 
  & $[17,9,7]_4$ & $(x^8 +\omega x^7 +\omega x^5 +\omega x^4 +\omega x^3 +\omega x + 1)$ & \\
35 & $[17,1,17]_4$ & $(x^{17}+1)/(x+1)$ & $[[17,12,4/2]]_4$ \\ 
  & $[17,13,4]_4$ & $(x^4 + x^3 + \omega^2 x^2 + x + 1)$ & \\
\hline
\end{tabular}
\end{table*}

\begin{table*}
\caption{Asymmetric QECC from Nested Cyclic Codes Continued}
\label{table:Cyclic2}
\centering
\begin{tabular}{| c | l | l | l |}
\hline
\textbf{No.} & \textbf{Codes $C$ and $D$} & \textbf{Generator Polynomials of $C$ and $D$} & \textbf{Code $Q$} \\
\hline
36 & $[19,9,8]_4$ & $(x+1)(x^9 +\omega^2 x^8 +\omega^2 x^6 +\omega^2 x^5 +\omega x^4 +\omega x^3 +\omega x + 1)$ & $[[19,1,7/7]]_4$\\ 
  & $[19,10,7]_4$ & $(x^9 + \omega^2 x^8 + \omega^2 x^6 + \omega^2 x^5 + \omega x^4 + \omega x^3 + \omega x + 1)$ & \\
37 & $[19,1,19]_4$ & $(x^{19}+1)/(x+1)$ & $[[19,9,7/2]]_4$ \\ 
  & $[19,10,7]_4$ & $(x^9 + \omega x^8 + \omega x^6 + \omega x^5 + \omega^2 x^4 + \omega^2 x^3 + \omega^2 x + 1)$ & \\
38 & $[21,1,21]_4$ & $(x^{21}+1)/(x+1)$ & $[[21,1,14/2]]_4$ \\ 
  & $[21,2,14]_4$ & $(x^{21}+1)/(x^2 + \omega^2 x + \omega)$ & \\
39 &$[21,4,12]_4$& $(x^{21}+1)/(x^4 + \omega x^3 + \omega^2 x^2 + x + 1)$ & $[[21,3,11/3]]_4$ \\ 
  &$[21,7,11]_4$& $x^{14} + \omega x^{13} + \omega^2 x^{12} + \omega^2 x^{10} + x^8 +\omega^2 x^7 + x^6 + \omega^2 x^4 + \omega^2 x^2 + \omega x + 1$ & \\
40 &$[21,4,12]_4$& $(x^{21}+1)/(x^4 + \omega^2 x^3 + \omega x^2 + x + 1)$  & $[[21,4,9/3]]_4$ \\
   &$[21,8,9]_4$&  $(x^{21}+1)/(x^8 + x^7 + \omega x^6 + \omega x^5 + x^4 + \omega^2 x^3 + x^2 + x + \omega)$ &\\

41 & $[21,1,21]_4$ & $(x^{21}+1)/(x+1)$ & $[[21,3,12/2]]_4$ \\ 
  & $[21,4,12]_4$ & $(x^{21}+1)/((x+1)(x^3+\omega^2 x^2+1))$ & \\
42 & $[21,7,11]_4$ & $(x^{21}+1)/(x^7 +\omega^2 x^6 + x^4 + x^3 +\omega^2 x + 1)$ & $[[21,3,8/5]]_4$ \\ 
  & $[21,10,8]_4$ & $(x^{11}+\omega x^{10}+ x^8 + x^7 +x^6 + x^5 +\omega x^4 +\omega x^3 +\omega^2 x^2 +\omega^2 x + 1)$ & \\
43 & $[21,4,12]_4$ & $(x^{21}+1)/(x^4 + x^3 + \omega x^2 +\omega^2 x + 1)$ & $[[21,4,9/3]]_4$ \\ 
  & $[21,8,9]_4$ & $(x^{21}+1)/(x^8 +\omega x^7 +\omega^2 x^6 + x^5 +\omega^2 x^4 +\omega x^3 +\omega x^2 +\omega)$ & \\
44 & $[21,7,11]_4$ & $(x^{21}+1)/(x^7 +\omega^2 x^6 + x^4 + x^3 +\omega^2 x + 1)$ & $[[21,4,6/5]]_4$ \\ 
  & $[21,11,6]_4$ & $(x^{21}+1)/(x^{11} + x^8 +\omega^2 x^7 + x^2 +\omega)$ & \\
45 & $[21,1,21]_4$ & $(x^{21}+1)/(x+1)$ & $[[21,6,11/2]]_4$ \\ 
  & $[21,7,11]_4$ & $(x^{21}+1)/(x^7 +\omega^2 x^6 + x^4 + x^3 +\omega^2 x + 1)$ & \\

46 & $[21,4,12]_4$ & $(x^{21}+1)/(x^4 + x^3 +\omega x^2 +\omega^2 x + 1)$ & $[[21,6,8/3]]_4$ \\ 
  & $[21,10,8]_4$ & $(x^{11}+\omega x^{10}+ x^8 + x^7 +x^6 + x^5 +\omega x^4 +\omega x^3 +\omega^2 x^2 +\omega^2 x + 1)$ & \\
47 &$[21,7,11]_4$& $(x^{21}+1)/(x^7 + \omega^2 x^6 + x^4 + x^3 + \omega ^2 x + 1)$ & $[[21,7,5/5]]_4$ \\ 
   &$[21,14,5]_4$& $x^7 + x^6 + x^4 + \omega x^3 + \omega^2x + \omega$ & \\
48 & $[21,1,21]_4$ & $(x^{21}+1)/(x+1)$ & $[[21,7,9/2]]_4$ \\ 
  & $[21,8,9]_4$ & $(x^{21}+1)/(x^8 +\omega x^7 +\omega^2 x^6 + x^5 +\omega^2 x^4 +\omega x^3 +\omega x^2 +\omega)$ & \\
49 & $[21,4,12]_4$ & $(x^{21}+1)/(x^4 + x^3 +\omega x^2 +\omega^2 x + 1)$ & $[[21,7,6/3]]_4$ \\ 
  & $[21,11,6]_4$ & $(x^{10} +\omega x^9 + x^8 +\omega x^7 + x^6 + x^5 + x^4 +\omega^2 x^2 +\omega^2)$ & \\
50 & $[21,1,21]_4$ & $(x^{21}+1)/(x+1)$ & $[[21,9,8/2]]_4$ \\ 
  & $[21,10,8]_4$ & $(x^{11}+\omega x^{10}+x^8+x^7+x^6+x^5+\omega x^4 +\omega x^3+\omega^2 x^2+\omega^2 x+ 1)$ & \\

51 & $[21,1,21]_4$ & $(x^{21}+1)/(x+1)$ & $[[21,10,6/2]]_4$ \\ 
  & $[21,11,6]_4$ & $(x^{10} + x^7 +\omega^2 x^6 + x^4 +\omega x^2 +\omega^2)$ & \\
52 & $[21,4,12]_4$ & $(x^{21}+1)/(x^4 + x^3 +\omega x^2 + \omega^2 x + 1)$ & $[[21,10,5/3]]_4$ \\ 
  & $[21,14,5]_4$ & $(x^7 + x^6 + x^4 +\omega x^3 +\omega^2 x +\omega)$ & \\
53 & $[21,1,21]_4$ & $(x^{21}+1)/(x+1)$ & $[[21,13,5/2]]_4$ \\ 
  & $[21,14,5]_4$ & $(x^7 + x^6 + x^4 +\omega x^3 +\omega^2 x +\omega)$ & \\
54 & $[21,1,21]_4$ & $(x^{21}+1)/(x+1)$ & $[[21,16,3/2]]_4$ \\ 
  & $[21,17,3]_4$ & $(x+\omega)(x^3+\omega^2 x^2 +1)$ & \\
55 & $[21,1,21]_4$ & $(x^{21}+1)/(x+1)$ & $\mathbf{[[21,19,2/2]]_4}$ \\ 
  & $[21,20,2]_4$ & $(x+\omega)$ & \\

56 & $[23,11,8]_4$ & $(x+1)(x^{11}+x^{10}+x^6+x^5+x^4+x^2+1)$ & $[[23,1,7/7]]_4$ \\ 
  & $[23,12,7]_4$ & $(x^{11}+x^{10}+x^6+x^5+x^4+x^2+1)$ & \\
57 & $[23,1,23]_4$ & $(x^{23}+1)/(x+1)$ & $[[23,11,7/2]]_4$ \\ 
  & $[23,12,7]_4$ & $(x^{11}+x^9+x^7+x^6+x^5+x+1)$ & \\
58 & $[25,1,25]_4$ & $(x^{25}+1)/(x+1)$ & $[[25,2,15/2]]_4$ \\ 
  & $[25,3,15]_4$ & $(x^{25}+1)/(x^3 + \omega x^2 + \omega x + 1)$ & \\
59 & $[25,12,4]_4$ & $(x^{13}+\omega^2 x^{12}+\omega^2 x^{11}+x^{10}+\omega x^8+x^7+x^6+\omega x^5 + x^3+\omega^2 x^2+\omega^2 x+ 1)$ & $[[25,2,4/4]]_4$ \\ 
  & $[25,14,4]_4$ & $(x^{11}+ x^{10}+\omega x^6 +\omega x^5 + x + 1)$ & \\
60 & $[25,1,25]_4$ & $(x^{25}+1)/(x+1)$ & $[[25,4,5/2]]_4$ \\ 
  & $[25,5,5]_4$ & $(x^{25}+1)/(x^5 + 1)$ & \\

61 & $[25,10,4]_4$ & $(x^{15}+\omega^2 x^{10} +\omega^2 x^5 + 1)$ & $[[25,4,4/3]]_4$ \\ 
  & $[25,14,4]_4$ & $(x^{11} + x^{10} +\omega x^6 +\omega x^5 + x + 1)$ & \\
62 & $[25,10,4]_4$ & $(x^{15}+\omega^2 x^{10} +\omega^2 x^5 + 1)$ & $[[25,5,3/3]]_4$ \\ 
  & $[25,15,3]_4$ & $(x^{10}+\omega x^5 +1)$ & \\
63 & $[25,1,25]_4$ & $(x^{23}+1)/(x + 1)$ & $[[25,12,4/2]]_4$ \\ 
  & $[25,13,4]_4$ & $(x^{12}+\omega x^{11}+ x^{10} +\omega^2 x^7 + x^6 +\omega^2 x^5 + x^2 +\omega x + 1)$ & \\
64 & $[25,1,25]_4$ & $(x^{23}+1)/(x+1)$ & $[[25,14,3/2]]_4$ \\ 
  & $[25,15,3]_4$ & $(x^{10}+\omega x^5 +1)$ & \\
65 & $[25,1,25]_4$ & $(x^{23}+1)/(x+1)$ & $[[25,22,2/2]]_4$ \\ 
  & $[25,23,2]_4$ & $(x^2 +\omega x +1)$ & \\
\hline
\end{tabular}
\end{table*}

\section{Construction from Nested Linear BCH Codes}\label{sec:BCHCodes}
It is well known (see~\cite[Sec. 3]{Cha98}) that finding the 
minimum distance or even finding a good lower bound on the 
minimum distance of a cyclic code is not a trivial problem. 
One important family of cyclic codes is the family of BCH codes. 
Their importance lies on the fact that their designed distance 
provides a reasonably good lower bound on the minimum distance. 
For more on BCH codes,~\cite[Ch. 5]{HP03} can be consulted.

The BCH Code constructor in MAGMA can be used to find nested 
codes to produce more asymmetric quantum codes. Table~\ref{table:BCH} 
lists down the BCH codes over $\F_{4}$ for $n=27$ to $n=51$ 
with $n$ coprime to $4$. For a fixed length $n$, the codes are 
nested, i.e., a code $C$ with dimension $k_{1}$ is a subcode of a 
code $D$ with dimension $k_{2} > k_{1}$. The construction process 
can be done for larger values of $n$ if so desired.

The range of the designed distances that can be supplied into MAGMA to 
come up with the code $C$ and the actual minimum distance of $C$ are 
denoted by $\delta(C)$ and $d(C)$, respectively. The minimum distance 
of $C^{\perp_{\tr}}$, which is needed in the computation of ${d_z,d_x}$, 
is denoted by $d(C^{\perp_{\tr}})$. To save space, the 
BCH $[n,1,n]_{4}$ repetition code generated by the all one vector $\1$ 
is not listed down in the table although this code is used in the 
construction of many asymmetric quantum codes presented 
in Table~\ref{table:BCH_QECC}.

\begin{table}
\caption{BCH Codes over $\F_{4}$ with $2 \leq k < n$ for $27\leq n \leq51$}
\label{table:BCH}
\centering
\begin{tabular}{| c | c | c | c | l | c |}
\hline
\textbf{No.} & \textbf{$n$} & \textbf{$\delta(C)$} & \textbf{$d(C)$} & \textbf{Code $C$} & \textbf{$d(C^{\perp_{\tr}})$} \\
\hline
1 & $27$ & $2$ & $2$ & $[27,18,2]_4$ & $3$\\ % inserting body of the table
2 &      & $3$ & $3$ & $[27,9,3]_4$  & $2$\\
3 &      & $4-6$ & $6$ & $[27,6,6]_4$ & $2$\\
4 &      & $7-9$ & $9$ & $[27,3,9]_4$ & $2$\\
5 &      & $10-18$ & $18$ & $[27,2,18]_4$ & $2$\\

6 & $29$ & $2$ & $11$ & $[29,15,11]_4$  & $12$\\

7 & $31$ & $2-3$ & $3$ & $[31,26,3]_4$ & $16$\\
8 &      & $4-5$ & $5$ & $[31,21,5]_4$ & $12$\\
9 &      & $6-7$ & $7$ & $[31,16,7]_4$ & $8$\\
10 &     & $8-11$ & $11$ & $[31,11,11]_4$ & $6$\\
11 &     & $12-15$ & $15$ & $[31,6,15]_4$ & $4$\\

12 & $33$ & $2$ & $2$ & $[33,28,2]_4$  & $18$\\
13 &      & $3$ & $3$ & $[33,23,3]_4$ & $12$\\
14 &      & $4-5$ & $8$ & $[33,18,8]_4$ & $11$\\
15 &      & $6$ & $10$ & $[33,13,10]_4$  & $6$\\
16 &      & $7$ & $11$ & $[33,8,11]_4$   & $4$\\
17 &      & $8-11$ & $11$ & $[33,3,11]_4$   & $2$\\
18 &      & $12-22$ & $22$ & $[33,2,22]_4$ & $2$\\

19 & $35$ & $2$ & $3$ & $[35,29,3]_4$ & $16$\\
20 &      & $3$ & $3$ & $[35,23,3]_4$ & $8$\\
21 &      & $4-5$ & $5$ & $[35,17,5]_4$ & $8$\\
22 &      & $6$ & $7$ & $[35,14,7]_4$  & $8$\\
23 &      & $7$ & $7$ & $[35,8,7]_4$  & $4$\\
24 &      & $8-14$ & $15$ & $[35,6,15]_4$ & $4$\\
25 &      & $15$ & $15$ & $[35,4,15]_4$ & $2$\\

26 & $37$ & $2$ & $11$ & $[37,19,11]_4$ & $12$\\

27 & $39$ & $2$ & $2$ & $[39,33,2]_4$ & $18$\\
28 &      & $3$ & $3$ & $[39,27,3]_4$ & $12$\\
29 &      & $4-6$ & $9$ & $[39,21,9]_4$ & $12$\\
30 &      & $7$ & $10$ & $[39,15,10]_4$  & $6$\\
31 &      & $8-13$ & $13$ & $[39,9,13]_4$ & $4$\\
32 &      & $14$ & $15$ & $[39,8,15]_4$ & $3$\\
33 &      & $15-26$ & $26$ & $[39,2,26]_4$ & $2$\\

34 & $41$ & $2$ & $6$ & $[41,31,6]_4$ & $20$\\
35 &      & $3$ & $9$ & $[41,21,9]_4$ & $10$\\
36 &      & $4-6$ & $20$ & $[41,11,20]_4$ & $7$\\

37 & $43$ & $2$ & $5$ & $[43,36,5]_4$ & $27$\\
38 &      & $3$ & $6$ & $[43,29,6]_4$ & $14$\\
39 &      & $4-6$ & $11$ & $[43,22,11]_4$ & $12$\\
40 &      & $7$ & $13$ & $[43,15,13]_4$ & $6$\\

41 &      & $8-9$ & $26$ & $[43,8,26]_4$ & $5$\\
42 & $45$ & $2$ & $2$ & $[45,39,2]_4$ & $12$\\
43 &      & $3$ & $3$ & $[45,33,3]_4$  & $8$\\
44 &      & $4-5$ & $5$ & $[45,31,5]_4$ & $8$\\
45 &      & $6$ & $6$ & $[45,28,6]_4$ & $8$\\

46 &      & $7$ & $7$ & $[45,26,7]_4$ & $8$\\
47 &      & $8-9$ & $9$ & $[45,20,9]_4$ & $6$\\
48 &      & $10$ & $10$ & $[45,18,10]_4$ & $6$\\
49 &      & $11$ & $11$ & $[45,15,11]_4$  & $3$\\

50 &      & $12-15$ & $15$ & $[45,9,15]_4$ & $2$\\
51 &      & $16-18$ & $18$ & $[45,8,18]_4$ & $2$\\
52 &      & $19-21$ & $21$ & $[45,6,21]_4$ & $2$\\
53 &      & $22-30$ & $30$ & $[45,4,30]_4$ & $2$\\
54 &      & $31-33$ & $33$ & $[45,3,33]_4$ & $2$\\

55 & $47$ & $2-5$ & $11$ & $[47,24,11]_4$ & $12$\\

56 & $49$ & $2-3$ & $3$ & $[49,28,3]_4$ & $4$\\
57 &      & $4-7$ & $7$ & $[49,7,7]_4$ & $2$\\
58 &      & $8-21$ & $21$ & $[49,4,21]_4$ & $2$\\

59 & $51$ & $2$ & $2$ & $[51,47,2]_4$  & $36$\\
60 &      & $3$ & $3$ & $[51,43,3]_4$ & $24$\\

61 &      & $4-5$ & $5$ & $[51,39,5]_4$ & $24$\\
62 &      & $6$ & $9$ & $[51,35,9]_4$ & $22$\\
63 &      & $7$ & $9$ & $[51,31,9]_4$ & $14$\\
64 &      & $8-9$ & $9$ & $[51,27,9]_4$ & $10$\\

65 &      & $10-11$ & $14$ & $[51,23,14]_4$ & $10$\\
66 &      & $12-17$ & $17$ & $[51,19,17]_4$ & $8$\\
67 &      & $18$ & $18$ & $[51,18,18]_4$ & $8$\\
68 &      & $19$ & $19$ & $[51,14,19]_4$ & $6$\\
69 &      & $20-22$ & $27$ & $[51,10,27]_4$ & $6$\\
70 &      & $23-34$ & $34$ & $[51,6,34]_4$ & $4$\\
71 &      & $35$ & $35$ & $[51,5,35]_4$ & $3$\\
\hline
\end{tabular}
\end{table}

Table~\ref{table:BCH_QECC} presents the resulting asymmetric 
quantum codes from nested BCH Codes based on Theorem~\ref{thm:3.5}. 
The inner codes are listed in the column denoted by Code 
$C_{1}^{\perp_{\tr}}$ while the corresponding larger codes are put 
in the column denoted by Code $C_{2}$. The values for $d_z,d_x$ are 
derived from the last column of Table~\ref{table:BCH} while keeping 
Proposition~\ref{prop:3.2} in mind.

\begin{table*}
\caption{Asymmetric QECC from BCH Codes}
\label{table:BCH_QECC}
\centering
\begin{tabular}{| c | c | l | l | l || c | c | l | l | l |}
\hline
\textbf{No.} & \textbf{$n$} & \textbf{Code $C_{1}^{\perp_{\tr}}$} & \textbf{Code $C_{2}$} & \textbf{Code $Q$} &
\textbf{No.} & \textbf{$n$} & \textbf{Code $C_{1}^{\perp_{\tr}}$} & \textbf{Code $C_{2}$} & \textbf{Code $Q$}\\
\hline
1 & $27$ & $[27,1,27]_4$ & $[27,2,18]_4$ & $[[27,1,18/2]]_4$ & 76 & $43$ & $[43,8,26]_4$ & $[43,29,6]_4$ & $[[43,21,6/5]]_4$\\
2 &      & $[27,1,27]_4$ & $[27,3,9]_4$ & $[[27,2,9/2]]_4$   & 77 &      & $[43,1,43]_4$ & $[43,29,6]_4$ & $[[43,28,6/2]]_4$\\
3 &      & $[27,1,27]_4$ & $[27,6,6]_4$ & $[[27,5,6/2]]_4$   & 78 &      & $[43,8,26]_4$ & $[43,36,5]_4$ & $[[43,28,5/5]]_4$\\
4 &      & $[27,1,27]_4$ & $[27,9,3]_4$ & $[[27,8,3/2]]_4$   & 79 &      & $[43,1,43]_4$ & $[43,36,5]_4$ & $[[43,35,5/2]]_4$\\
5 &      & $[27,1,27]_4$ & $[27,18,2]_4$ & $[[27,17,2/2]]_4$ & 80 & $45$ & $[45,1,45]_4$ & $[45,3,33]_4$ & $[[45,2,33/2]]_4$\\

6 & $29$ & $[29,1,29]_4$ & $[29,15,11]_4$ & $[[29,14,11/2]]_4$ & 81 &      & $[45,18,10]_4$ & $[45,20,9]_4$ & $[[45,2,9/6]]_4$\\ 
7 & $31$ & $[31,1,31]_4$ & $[31,6,15]_4$ & $[[31,5,15/2]]_4$  & 82 &      & $[45,1,45]_4$ & $[45,4,30]_4$ & $[[45,3,30/2]]_4$\\
8 &      & $[31,21,5]_4$ & $[31,26,3]_4$ & $[[31,5,12/3]]_4$  & 83 &      & $[45,15,11]_4$ & $[45,18,10]_4$ & $[[45,3,10/3]]_4$\\
9 &      & $[31,6,15]_4$ & $[31,11,11]_4$ & $[[31,5,11/4]]_4$ & 84 &      & $[45,1,45]_4$ & $[45,6,21]_4$ & $[[45,5,21/2]]_4$\\
10 &     & $[31,16,7]_4$ & $[31,21,5]_4$ & $[[31,5,8/5]]_4$   & 85 &      & $[45,15,11]_4$ & $[45,20,9]_4$ & $[[45,5,9/3]]_4$\\

11 &     & $[31,11,11]_4$ & $[31,16,7]_4$ & $[[31,5,7/6]]_4$  & 86 &      & $[45,26,7]_4$ & $[45,31,5]_4$ & $[[45,5,8/5]]_4$\\
12 &     & $[31,1,31]_4$ & $[31,11,11]_4$ & $[[31,10,11/2]]_4$  & 87 &      & $[45,1,45]_4$ & $[45,8,18]_4$ & $[[45,7,18/2]]_4$\\
13 &     & $[31,16,7]_4$ & $[31,26,3]_4$ & $[[31,10,8/3]]_4$  & 88 &      & $[45,26,7]_4$ & $[45,33,3]_4$ & $[[45,7,8/3]]_4$\\
14 &     & $[31,6,15]_4$ & $[31,16,7]_4$ & $[[31,10,7/4]]_4$  & 89 &      & $[45,1,45]_4$ & $[45,9,15]_4$ & $[[45,8,15/2]]_4$\\
15 &     & $[31,11,11]_4$ & $[31,21,5]_4$ & $[[31,10,6/5]]_4$ & 90 &      & $[45,18,10]_4$ & $[45,26,7]_4$ & $[[45,8,7/6]]_4$\\

16 &     & $[31,1,31]_4$ & $[31,16,7]_4$ & $[[31,15,7/2]]_4$  & 91 &      & $[45,18,10]_4$ & $[45,28,6]_4$ & $[[45,10,6/6]]_4$\\
17 &     & $[31,11,11]_4$ & $[31,26,3]_4$ & $[[31,15,6/3]]_4$ & 92 &      & $[45,15,11]_4$ & $[45,26,7]_4$ & $[[45,11,7/3]]_4$\\
18 &     & $[31,6,15]_4$ & $[31,21,5]_4$ & $[[31,15,5/4]]_4$  & 93 &      & $[45,18,10]_4$ & $[45,31,5]_4$ & $[[45,13,6/5]]_4$\\
19 &     & $[31,1,31]_4$ & $[31,21,5]_4$ & $[[31,20,5/2]]_4$  & 94 &      & $[45,1,45]_4$ & $[45,15,11]_4$ & $[[45,14,11/2]]_4$\\
20 &     & $[31,6,15]_4$ & $[31,26,3]_4$ & $[[31,20,4/3]]_4$  & 95 &      & $[45,18,10]_4$ & $[45,33,3]_4$ & $[[45,15,6/3]]_4$\\

21 &     & $[31,1,31]_4$ & $[31,26,3]_4$ & $[[31,25,3/2]]_4$  & 96 &      & $[45,15,11]_4$ & $[45,31,5]_4$ & $[[45,16,5/3]]_4$\\
22 & $33$ & $[33,1,33]_4$ & $[33,2,22]_4$ & $[[33,1,22/2]]_4$ & 97 &      & $[45,1,45]_4$ & $[45,18,10]_4$ & $[[45,17,10/2]]_4$\\
23 &      & $[33,23,3]_4$ & $[33,28,2]_4$ & $[[33,5,12/2]]_4$ & 98 &    & $[45,15,11]_4$ & $[45,33,3]_4$ & $[[45,18,3/3]]_4$\\
24 &      & $[33,18,8]_4$ & $[33,23,3]_4$ & $[[33,5,11/3]]_4$ & 99 &      & $[45,1,45]_4$ & $[45,20,9]_4$ & $[[45,19,9/2]]_4$\\
25 &      & $[33,8,11]_4$ & $[33,13,10]_4$ & $[[33,5,10/4]]_4$ & 100 &      & $[45,1,45]_4$ & $[45,26,7]_4$ & $[[45,25,7/2]]_4$\\

26 &      & $[33,13,10]_4$ & $[33,18,8]_4$ & $[[33,5,8/6]]_4$ & 101 &      & $[45,1,45]_4$ & $[45,28,6]_4$ & $[[45,27,6/2]]_4$\\
27 &      & $[33,18,8]_4$ & $[33,28,2]_4$ & $[[33,10,11/2]]_4$ & 102 &      & $[45,1,45]_4$ & $[45,31,5]_4$ & $[[45,30,5/2]]_4$\\
28 &      & $[33,8,11]_4$ & $[33,18,8]_4$ & $[[33,10,8/4]]_4$ & 103 &      & $[45,1,45]_4$ & $[45,33,3]_4$ & $[[45,32,3/2]]_4$\\
29 &      & $[33,1,33]_4$ & $[33,13,10]_4$ & $[[33,12,10/2]]_4$ & 104 &      & $[45,1,45]_4$ & $[45,39,2]_4$ & $[[45,38,2/2]]_4$\\
30 &      & $[33,8,11]_4$ & $[33,23,3]_4$ & $[[33,15,4/3]]_4$ & 105 & $47$ & $[47,1,47]_4$ & $[47,24,11]_4$ & $[[47,23,11/2]]_4$\\

31 &      & $[33,1,33]_4$ & $[33,18,8]_4$ & $[[33,17,8/2]]_4$ & 106 & $49$ & $[49,1,49]_4$ & $[49,4,21]_4$ & $[[49,3,21/2]]_4$\\
32 &      & $[33,8,11]_4$ & $[33,28,2]_4$ & $[[33,20,4/2]]_4$ & 107 &      & $[49,1,49]_4$ & $[49,7,7]_4$ & $[[49,6,7/2]]_4$\\
33 &      & $[33,1,33]_4$ & $[33,23,3]_4$ & $[[33,22,3/2]]_4$ & 108 &      & $[49,1,49]_4$ & $[49,28,3]_4$ & $[[49,27,3/2]]_4$\\
34 &      & $[33,1,33]_4$ & $[33,28,2]_4$ & $[[33,27,2/2]]_4$ & 109 & $51$ & $[51,5,35]_4$ & $[51,6,34]_4$ & $[[51,1,34/3]]_4$\\
35 & $35$ & $[35,14,7]_4$ & $[35,17,5]_4$ & $[[35,3,8/5]]_4$  & 110 &      & $[51,18,18]_4$ & $[51,19,17]_4$ & $[[51,1,17/8]]_4$\\

36 &      & $[35,1,35]_4$ & $[35,6,15]_4$ & $[[35,5,15/2]]_4$ & 111 &      & $[51,1,51]_4$ & $[51,5,35]_4$ & $[[51,4,35/2]]_4$\\
37 &      & $[35,6,15]_4$ & $[35,14,7]_4$ & $[[35,8,7/4]]_4$  & 112 &      & $[51,6,34]_4$ & $[51,10,27]_4$ & $[[51,4,27/4]]_4$\\
38 &      & $[35,6,15]_4$ & $[35,17,5]_4$ & $[[35,11,5/4]]_4$  & 113 &      & $[51,10,27]_4$ & $[51,14,19]_4$ & $[[51,4,19/6]]_4$\\
39 &      & $[35,14,7]_4$ & $[35,29,3]_4$ & $[[35,15,8/3]]_4$ & 114 &      & $[51,1,51]_4$ & $[51,6,34]_4$ & $[[51,5,34/2]]_4$\\
40 &      & $[35,1,35]_4$ & $[35,17,5]_4$ & $[[35,16,5/2]]_4$ & 115 &      & $[51,5,35]_4$ & $[51,10,27]_4$ & $[[51,5,27/3]]_4$\\

41 &      & $[35,6,15]_4$ & $[35,29,3]_4$ & $[[35,23,4/3]]_4$ & 116 &      & $[51,18,18]_4$ & $[51,23,14]_4$ & $[[51,5,14/8]]_4$\\
42 &      & $[35,1,35]_4$ & $[35,29,3]_4$ & $[[35,28,3/2]]_4$ & 117 &      & $[51,39,5]_4$ & $[51,47,2]_4$ & $[[51,8,24/2]]_4$\\
43 & $37$ & $[37,1,37]_4$ & $[37,19,2]_4$ & $[[37,18,11/2]]_4$ & 118 &      & $[51,6,34]_4$ & $[51,14,19]_4$ & $[[51,8,19/4]]_4$\\
44 & $39$ & $[39,1,39]_4$ & $[39,2,26]_4$ & $[[39,1,26/2]]_4$ & 119 &      & $[51,10,27]_4$ & $[51,18,18]_4$ & $[[51,8,18/6]]_4$\\
45 &      & $[39,1,39]_4$ & $[39,2,26]_4$ & $[[39,1,26/2]]_4$ & 120 &      & $[51,1,51]_4$ & $[51,10,27]_4$ & $[[51,9,27/2]]_4$\\

46 &      & $[39,8,15]_4$ & $[39,9,13]_4$ & $[[39,1,13/3]]_4$ & 121 &      & $[51,5,35]_4$ & $[51,14,19]_4$ & $[[51,9,19/3]]_4$\\
47 &      & $[39,21,9]_4$ & $[39,27,3]_4$ & $[[39,6,12/3]]_4$ & 122 &      & $[51,10,27]_4$ & $[51,19,17]_4$ & $[[51,9,17/6]]_4$\\
48 &      & $[39,9,13]_4$ & $[39,15,10]_4$ & $[[39,6,10/4]]_4$ & 123 &      & $[51,6,34]_4$ & $[51,18,18]_4$ & $[[51,12,18/4]]_4$\\
49 &      & $[39,15,10]_4$ & $[39,21,9]_4$ & $[[39,6,9/6]]_4$ & 124 &      & $[51,1,51]_4$ & $[51,14,19]_4$ & $[[51,13,19/2]]_4$\\
50 &      & $[39,1,39]_4$ & $[39,8,15]_4$ & $[[39,7,15/2]]_4$ & 125 &      & $[51,5,35]_4$ & $[51,18,18]_4$ & $[[51,13,18/3]]_4$\\

51 &      & $[39,8,15]_4$ & $[39,15,10]_4$ & $[[39,7,10/3]]_4$ & 126 &      & $[51,6,34]_4$ & $[51,19,17]_4$ & $[[51,13,17/4]]_4$\\
52 &      & $[39,1,39]_4$ & $[39,9,13]_4$ & $[[39,8,13/2]]_4$ & 127 &      & $[51,10,27]_4$ & $[51,23,14]_4$ & $[[51,13,14/6]]_4$\\
53 &      & $[39,21,9]_4$ & $[39,33,2]_4$ & $[[39,12,12/2]]_4$ & 128 &      & $[51,5,35]_4$ & $[51,19,17]_4$ & $[[51,14,17/3]]_4$\\
54 &      & $[39,9,13]_4$ & $[39,21,9]_4$ & $[[39,12,9/4]]_4$  & 129 &      & $[51,1,51]_4$ & $[51,18,18]_4$ & $[[51,17,18/2]]_4$\\
55 &      & $[39,8,15]_4$ & $[39,21,9]_4$ & $[[39,13,9/3]]_4$  & 130 &      & $[51,6,34]_4$ & $[51,23,14]_4$ & $[[51,17,14/4]]_4$\\

56 &      & $[39,1,39]_4$ & $[39,15,10]_4$ & $[[39,14,10/2]]_4$  & 131 &      & $[51,18,18]_4$ & $[51,35,9]_4$ & $[[51,17,9/8]]_4$\\
57 &      & $[39,9,13]_4$ & $[39,27,3]_4$ & $[[39,18,4/3]]_4$  & 132 &      & $[51,1,51]_4$ & $[51,19,17]_4$ & $[[51,18,17/2]]_4$\\
58 &      & $[39,8,15]_4$ & $[39,27,3]_4$ & $[[39,19,3/3]]_4$  & 133 &      & $[51,5,35]_4$ & $[51,23,14]_4$ & $[[51,18,14/3]]_4$\\ 
%entry 44 non self-orthog codes
59 &      & $[39,1,39]_4$ & $[39,21,9]_4$ & $[[39,20,9/2]]_4$  & 134 &      & $[51,1,51]_4$ & $[51,23,14]_4$ & $[[51,22,14/2]]_4$\\
60 &      & $[39,9,13]_4$ & $[39,33,2]_4$ & $[[39,24,4/2]]_4$  & 135 &      & $[51,10,27]_4$ & $[51,35,9]_4$ & $[[51,25,9/6]]_4$\\ 

61 &      & $[39,1,39]_4$ & $[39,27,3]_4$ & $[[39,26,3/2]]_4$  & 136 &      & $[51,6,34]_4$ & $[51,35,9]_4$ & $[[51,29,9/4]]_4$\\
62 &      & $[39,1,39]_4$ & $[39,33,2]_4$ & $[[39,32,2/2]]_4$  & 137 &      & $[51,10,27]_4$ & $[51,39,5]_4$ & $[[51,29,6/5]]_4$\\
63 & $41$ & $[41,1,41]_4$ & $[41,11,20]_4$ & $[[41,10,20/2]]_4$  & 138 &      & $[51,5,35]_4$ & $[51,35,9]_4$ & $[[51,30,9/3]]_4$\\
64 &      & $[41,21,9]_4$ & $[41,31,6]_4$ & $[[41,10,10/6]]_4$ & 139 &      & $[51,10,27]_4$ & $[51,43,3]_4$ & $[[51,33,6/3]]_4$\\
65 &      & $[41,11,20]_4$ & $[41,21,9]_4$ & $[[41,10,9/7]]_4$ & 140 &      & $[51,6,34]_4$ & $[51,39,5]_4$ & $[[51,33,5/4]]_4$\\

66 &      & $[41,1,41]_4$ & $[41,21,9]_4$ & $[[41,20,9/2]]_4$   & 141 &      & $[51,1,51]_4$ & $[51,35,9]_4$ & $[[51,34,9/2]]_4$\\
67 &      & $[41,11,20]_4$ & $[41,31,6]_4$ & $[[41,20,7/6]]_4$  & 142 &      & $[51,5,35]_4$ & $[51,39,5]_4$ & $[[51,34,5/3]]_4$\\
68 &      & $[41,1,41]_4$ & $[41,31,6]_4$ & $[[41,30,6/2]]_4$  & 143 &      & $[51,10,27]_4$ & $[51,47,2]_4$ & $[[51,37,6/2]]_4$\\
69 & $43$ & $[43,1,43]_4$ & $[43,8,26]_4$ & $[[43,7,26/2]]_4$  & 144 &      & $[51,6,34]_4$ & $[51,43,3]_4$ & $[[51,37,4/3]]_4$\\
70 &      & $[43,29,6]_4$ & $[43,36,5]_4$ & $[[43,7,14/5]]_4$ & 145 &      & $[51,1,51]_4$ & $[51,39,5]_4$ & $[[51,38,5/2]]_4$\\

71 &      & $[43,22,11]_4$ & $[43,29,6]_4$ & $[[43,7,12/6]]_4$ & 146 &     & $[51,5,35]_4$ & $[51,43,3]_4$ & $[[51,38,3/3]]_4$\\
72 & $43$ & $[43,1,43]_4$ & $[43,15,13]_4$ & $[[43,14,13/2]]_4$ & 147 &      & $[51,6,34]_4$ & $[51,47,2]_4$ & $[[51,41,4/2]]_4$\\
73 &      & $[43,22,11]_4$ & $[43,36,5]_4$ & $[[43,14,12/5]]_4$ & 148 &      & $[51,5,35]_4$ & $[51,47,2]_4$ & $[[51,42,3/2]]_4$\\
74 &      & $[43,15,13]_4$ & $[43,29,6]_4$ & $[[43,14,6/6]]_4$ & 149 &      & $[51,1,51]_4$ & $[51,47,2]_4$ & $[[51,46,2/2]]_4$\\
75 &      & $[43,1,43]_4$ & $[43,22,11]_4$ & $[[43,21,11/2]]_4$ &    &      &               &               & \\
\hline
\end{tabular}
\end{table*}
\section{Asymmetric Quantum Codes from Nested Additive Codes over $\F_4$}\label{sec:nestedadditive}

To show the gain that we can get from Theorem~\ref{thm:3.5} over the 
construction which is based solely on $\F_{4}$-linear codes, we 
exhibit asymmetric quantum codes which are derived from nested 
additive codes. 

An example of asymmetric quantum code with $k>0$ can be 
derived from a self-orthogonal additive cyclic code listed as 
Entry 3 in~\cite[Table I]{CRSS98}. The code is of parameters 
$(21,2^{16},9)_4$ yielding a $[[21,5,6/6]]_4$ quantum code 
$Q$ by Theorem~\ref{thm:3.5}. In a similar manner, a $[[23,12,4/4]]_4$ 
quantum code can be derived from Entry 5 of the same table.

Another very interesting example is the $(12,2^{12},6)_4$ dodecacode 
$C$ mentioned in Remark~\ref{rem:5.3}. Its generator matrix $G$ is 
given in Equation (\ref{eq:9.1}).

Let $G_{D},G_{E}$ be matrices formed, respectively, by deleting the last 
4 and 8 rows of $G$. Construct two additive codes $D,E \subset C$ with 
generator matrices $G_{D}$ and $G_{E}$, respectively. Applying 
Theorem~\ref{thm:3.5} with $C_1=D^{\perp_{\tr}}$ and $C_2 = C$ yields 
an asymmetric quantum code $Q$ with parameters $[[12,2,6/3]]_4$. 
Performing the same process to $E \subset C$ results in a 
$[[12,4,6/2]]_4$-code.
\begin{equation}\label{eq:9.1}
G=\left(
\begin{array}{*{12}{l}}
0 & 0 & 0 & 0 & 0 & 0 & 1 & 1 & 1 & 1 & 1 & 1\\
0 & 0 & 0 & 0 & 0 & 0 & \omega & \omega & \omega & \omega & \omega & \omega\\
1 & 1 & 1 & 1 & 1 & 1 & 0 & 0 & 0 & 0 & 0 & 0 \\

\omega & \omega & \omega & \omega & \omega & \omega & 0 & 0 & 0 & 0 & 0 & 0 \\
0 & 0 & 0 & 1 & \omega & \overline{\omega} & 0 & 0 & 0 & 1 & \omega & \overline{\omega}\\
0 & 0 & 0 & \omega & \overline{\omega} & 1 & 0 & 0 & 0 & \omega & \overline{\omega} & 1\\

1 & \overline{\omega} & \omega & 0 & 0 & 0 & 1 & \overline{\omega} & \omega & 0 & 0 & 0\\
\omega & 1 & \overline{\omega} & 0 & 0 & 0 & \omega & 1 & \overline{\omega} & 0 & 0 & 0\\
0 & 0 & 0 & 1 & \overline{\omega} & \omega & \omega & \overline{\omega} & 1 & 0 & 0 & 0\\

0 & 0 & 0 & \omega & 1 & \overline{\omega} & 1 & \omega & \overline{\omega} & 0 & 0 & 0\\
1 & \omega & \overline{\omega} & 0 & 0 & 0 & 0 & 0 & 0 & \overline{\omega} & \omega & 1\\
\overline{\omega} & 1 & \omega & 0 & 0 & 0 & 0 & 0 & 0 & 1 & \overline{\omega} & \omega
\end{array}
\right)\text{.}
\end{equation}

The next three subsections present more systematic approaches to finding 
good asymmetric quantum codes based on nested additive codes.

\subsection{Construction from circulant codes}\label{subsec:circulant}
As is the case with linear codes, an additive code $C$ is said to be cyclic 
if, given a codeword $\v \in C$, the cyclic shift of $\v$ is also in $C$. 
It is known (see~\cite[Th. 14]{CRSS98}) that any additive cyclic 
$(n,2^{k})_{4}$-code $C$ has at most two generators. A more detailed 
study of additive cyclic codes over $\F_{4}$ is given in~\cite{Huff07}. 

Instead of using additive cyclic codes, a subfamily which is called 
\textit{additive circulant $(n,2^{n})_4$-code} in~\cite{GK04} is used for ease of 
computation. An additive circulant code $C$ has as a generator matrix $G$ 
the complete cyclic shifts of just one codeword $\v=(v_{1},v_{2},\ldots,v_{n})$. 
We call $G$ the \textit{cyclic development of $\v$}.
More explicitly, $G$ is given by
\begin{equation}\label{eq:circ}
G=\left(
\begin{array}{*{12}{l}}
v_{1}  & v_{2}  & v_{3}  & \ldots & v_{n-1} & v_{n}\\
v_{n}  & v_{1}  & v_{2}  & \ldots & v_{n-2} & v_{n-1}\\
\vdots & \vdots & \vdots & \ddots & \vdots  & \vdots \\
v_{2}  & v_{3}  & v_{4}  & \ldots & v_{n}   & v_{1}
\end{array}
\right)\text{.}
\end{equation}

To generate a subcode of a circulant extremal self-dual code $C$ 
we delete the rows of its generator matrix $G$ starting from the last row, 
the first row being the generating codeword $\v$. We record the best 
possible combinations of the size of the resulting code $Q$ and 
$\left\lbrace d_z,d_x \right\rbrace$. To save space, only new codes 
or codes with better parameters than those previously constructed 
are presented.

Table~\ref{table:Circu} summarizes the finding for $n \leq 30$. Zeros on the 
right of each generating codeword are omitted. The number of last rows to 
be deleted to obtain the desired subcode is given in the column 
denoted by \textbf{del}.

\begin{table*}
\caption{Asymmetric Quantum Codes from Additive Circulant Codes for $n \leq 30$}
\label{table:Circu}
\centering
\begin{tabular}{| c | c | l | c | l |}
\hline
\textbf{No.} & \textbf{$n$} & \textbf{Generator $\v$} & \textbf{del} & \textbf{Code $Q$} \\
\hline
1 & $8$ & $(\overline{\omega},1,\omega,0,1)$ & $1$ & $[[8,0.5,4/2]]_4$ \\ % inserting body of the table
2 & $10$ & $(\overline{\omega},1,\omega,\omega,0,0,1)$ & $2$ & $[[10,1,4/3]]_4$ \\
3 &      &                                             & $3$ & $[[10,1.5,4/2]]_4$ \\
4 & $12$ & $(1,0,\omega,1,\omega,0,1)$                 & $2$ & $[[12,1,5/3]]_4$ \\
5 & $14$ & $(\omega,1,1,\overline{\omega},0,\omega,0,1)$ & $1$ & $[[14,0.5,6/3]]_4$ \\

6 &      &                                              & $5$ & $[[14,2.5,6/2]]_4$ \\
7 & $16$ & $(\omega,1,\overline{\omega},\omega,0,0,0,\omega,1)$ & $2$ & $[[16,1,6/4]]_4$ \\
8 & $16$ & $(\omega,1,1,0,0,\overline{\omega},\omega,0,1)$ & $6$ & $[[16,3,6/3]]_4$ \\
9 & $16$ & $(\omega,1,\overline{\omega},\omega,0,0,0,\omega,1)$ & $7$ & $[[16,3.5,6/2]]_4$ \\
10 & $19$ & $(1,0,\omega,\overline{\omega},1,\overline{\omega},\omega,0,1)$ & $4$ & $[[19,2,7/4]]_4$ \\

11 & $20$ & $(\overline{\omega},\overline{\omega},\omega,\overline{\omega},\omega,1,0,0,0,1,1)$ & $2$ & $[[20,1,8/5]]_4$ \\
12 &      &                                                                                     & $3$ & $[[20,1.5,8/4]]_4$ \\
13 &      &                                                                                     & $7$ & $[[20,3.5,8/3]]_4$ \\
14 & $22$ & $(\omega,\omega,\omega,\omega,1,1,\overline{\omega},\omega,0,\omega,\omega,0,\omega,\omega,1,1)$ & $4$&$[[22,2,8/5]]_4$\\
15 &      &                                                                                     & $6$ & $[[22,3,8/4]]_4$ \\

16 &      &                                                                                     & $10$ & $[[22,5,8/3]]_4$ \\
17 &      &                                                                                     & $11$ & $[[22,5.5,8/2]]_4$ \\
18 & $23$ & $(1,\omega,\omega,1,\overline{\omega},1,\omega,\omega,1)$ & $2$ & $[[23,1,8/4]]_4$\\
19 &      &                                                           & $6$ & $[[23,3,8/3]]_4$ \\
20 & $25$ & $(1,1,\omega,0,1,\overline{\omega},1,0,\omega,1,1)$       & $3$ & $[[25,1.5,8/5]]_4$\\

21 &      &                                                           & $6$ & $[[25,3,8/4]]_4$ \\
%done nicely
%for n=26 and 27 only one code each has been considered. need relooking
22 & $26$ & $(1,0,\omega,\omega,\omega,\overline{\omega},\omega,\omega,\omega,0,1)$ & $5$ & $[[26,2.5,8/4]]_4$\\
23 &      &                                                                         & $6$ & $[[26,3,8/3]]_4$ \\
%for 27
24 & $27$ & $(1,0,\omega,1,\omega,\overline{\omega},\omega,1,\omega,0,1)$ & $1$ & $[[27,0.5,8/5]]_4$\\
25 & $27$ & $(1,0,\omega,\omega,1,\overline{\omega},1,\omega,\omega,0,1)$ & $5$ & $[[27,2.5,8/4]]_4$ \\
26 &      &                                                               & $6$ & $[[27,3,8/3]]_4$\\
%for 28
27 & $28$ & $(\overline{\omega},\omega,\overline{\omega},1,\overline{\omega},1,\omega,\omega,\overline{\omega},\overline{\omega},\omega,\omega,0,1,1)$ & $1$ & $[[28,0.5,10/7]]_4$\\
28 &      &                                                                                     & $2$ & $[[28,1,10/5]]_4$ \\
29 &      &                                                                                     & $9$ & $[[28,4.5,10/4]]_4$ \\
30 &      &                                                                                     & $11$ & $[[28,5.5,10/3]]_4$ \\
%for 29
31 & $29$ & $(1,\omega,0,\omega,\overline{\omega},1,\overline{\omega},\omega,\overline{\omega},1,\overline{\omega},\omega,0,\omega,1)$ & $1$ & $[[29,0.5,11/7]]_4$\\
32 &      &                                                                                     & $3$ & $[[29,1.5,11/6]]_4$ \\
33 &      &                                                                                     & $8$ & $[[29,4,11/4]]_4$ \\
34 &      &                                                                                     & $12$ & $[[29,6,11/3]]_4$ \\
%for 30
35 & $30$ & $(\overline{\omega},0,\overline{\omega},\omega,1,\omega,0,\overline{\omega},\omega,0,1,\omega,1,1,0,1)$ & $5$ & $[[30,2.5,12/6]]_4$\\
36 &      &                                                                                     & $6$ & $[[30,3,12/5]]_4$ \\
37 &      &                                                                                     & $10$ & $[[30,5,12/3]]_4$ \\
38 &      &                                                                                     & $11$ & $[[30,5.5,12/2]]_4$ \\
\hline
\end{tabular}
\end{table*}

\subsection{Construction from $4$-circulant and bordered $4$-circulant codes}\label{subsec:4circulant}
Following~\cite{GK04}, a \textit{$4$-circulant additive $(n,2^{n})_4$-code 
of even length $n$} has the following generator matrix: 
\begin{equation}\label{eq:4circ}
G=\left(
\begin{array}{*{12}{l}}
I_{\frac{n}{2}}  & A_{\frac{n}{2}}\\
B_{\frac{n}{2}}  & I_{\frac{n}{2}}
\end{array}
\right)
\end{equation}
where $I_{\frac{n}{2}}$ is an identity matrix of size $n/2$ and 
$A_{\frac{n}{2}},B_{\frac{n}{2}}$ are circulant matrices of the form 
given in Equation (\ref{eq:circ}).

Starting from a generator matrix $G_{C}$ of an additive $4$-circulant code $C$, 
a matrix $G_{D}$ is constructed by deleting the last $r$ rows of $G_{C}$ to 
derive an additive subcode $D$ of $C$. For $n \leq 30$ we found three asymmetric quantum codes 
which are either new or better than the ones previously constructed. 
Table~\ref{table:4Circu} presents the findings. 
Under the column denoted by \textbf{$A,B$} we list down the generating 
codewords for the matrices $A$ and $B$, in that order.

\begin{table}
\caption{Asymmetric Quantum Codes from Additive 4-Circulant Codes for $n \leq 30$}
\label{table:4Circu}
\centering
\begin{tabular}{| c | c | c | l |}
\hline
\textbf{$n$} & \textbf{$A,B$} & \textbf{del} & \textbf{Code $Q$} \\
\hline
%$12$ & $(\overline{\omega},\overline{\omega},1,\omega,0,0)$,  & $3$ & $[[12,1.5,5/3]]_4$ \\
%     & $(\overline{\omega},1,0,\omega,0,\omega)$              &     &         \\
% this one is inferior to the one from dodecacode
$14$ & $(1,\omega,\omega,\omega,1,0,0)$,                      & $2$ & $[[14,1,6/3]]_4$ \\
     & $(1,0,0,1,\omega,\omega,\omega)$                       &     &         \\
$20$ & $(\overline{\omega},\omega,\omega,\omega,\omega,\omega,\overline{\omega},0,\omega,0)$,   & $8$ & $[[20,4,8/2]]_4$ \\
     & $(\omega,0,\overline{\omega},0,\omega,\omega,\overline{\omega},\omega,\overline{\omega},\omega)$ & & \\
\hline
\end{tabular}
\end{table}

Let $\d=(\omega,\ldots,\omega)$ and $\c$ be the transpose of $\d$. 
A \textit{bordered $4$-circulant additive $(n,2^{n})_4$-code 
of odd length $n$} has the following generator matrix: 
\begin{equation}\label{eq:bordered4circ}
G=\left(
\begin{array}{*{12}{l}}
e               & \1                 & \d\\
\1              & I_{\frac{n-1}{2}}  & A_{\frac{n-1}{2}}\\
\c              & B_{\frac{n-1}{2}}  & I_{\frac{n-1}{2}}
\end{array}
\right)\
\end{equation}
where $e$ is one of $0,1,\omega$, or $\overline{\omega}$, and 
$A_{\frac{n-1}{2}},B_{\frac{n-1}{2}}$ are circulant matrices.

We perform the procedure of constructing a subcode $D$ of $C$ by 
deleting the rows of $G_{C}$, starting from the last row. For 
$n \leq 30$, the five asymmetric quantum codes, either new or of 
better parameters, found can be seen in Table~\ref{table:B4Circu}. 
As before, under the column denoted by \textbf{$A,B$} we list 
down the generating codewords for the matrices $A$ and $B$, 
in that order.

\begin{table}
\caption{Asymmetric Quantum Codes from Additive Bordered 4-Circulant Codes for $n \leq 30$}
\label{table:B4Circu}
\centering
\begin{tabular}{| c | c | c | c | l |}
\hline
\textbf{$n$} & \textbf{$e$} & \textbf{$A,B$} & \textbf{del} & \textbf{Code $Q$} \\
\hline
$23$ & $\omega$ & $(\overline{\omega},1,\overline{\omega},\omega,\omega,1,1,0,0,0,0)$,                & $1$ & $[[23,0.5,8/5]]_4$ \\
     &          & $(0,\overline{\omega},\omega,\omega,\omega,\omega,\overline{\omega},1,0,1,\omega)$ &     &         \\
$23$ & $\omega$ & $(\overline{\omega},1,\overline{\omega},\omega,\omega,1,1,0,0,0,0)$,                & $4$ & $[[23,2,8/4]]_4$ \\
     &          & $(0,\overline{\omega},\omega,\omega,\omega,\omega,\overline{\omega},1,0,1,\omega)$ &     &         \\
$23$ & $\omega$ & $(1,1,\overline{\omega},1,\overline{\omega},1,1,0,0,0,0)$,         & $9$ & $[[23,4.5,8/2]]_4$ \\
     &   &$(\overline{\omega},\overline{\omega},\omega,\omega,\overline{\omega},\overline{\omega},\omega,0,\omega,0,\omega)$ & & \\
$25$ & $\omega$ & $(\omega,1,1,\omega,\overline{\omega},1,0,1,0,0,0,0)$,         & $4$ & $[[25,2,8/5]]_4$ \\
   & &
$(1,\overline{\omega},\overline{\omega},\omega,\omega,\overline{\omega},\omega,\overline{\omega},1,0,\overline{\omega},\omega)$ & & \\
$25$ & $\omega$ & $(\omega,1,\overline{\omega},1,1,\omega,0,1,0,0,0,0)$,         & $10$ & $[[25,5,8/2]]_4$ \\
   &  &
$(1,\overline{\omega},\overline{\omega},\overline{\omega},\omega,\overline{\omega},\omega,1,\omega,\overline{\omega},0,\omega)$ & & \\
\hline
\end{tabular}
\end{table}

\begin{rem}\label{rem:9.1}
A similar procedure has been done to the generator matrices of 
\textit{s-extremal additive codes} found in~\cite{BGKW07} 
and~\cite{Var09} as well as to the \textit{formally self-dual 
additive codes} of~\cite{HK08}. So far we have found no new or 
better asymmetric codes from these sources.
\end{rem}

Deleting the rows of $G_{C}$ in a more careful way than just doing so 
consecutively starting from the last row may yield new or better 
asymmetric quantum codes. The process, however, is more time consuming.

Consider the following instructive example taken from bordered 
4-circulant codes. Let 
\begin{equation*}
P:=\left\lbrace 1,2,4,5,8,10,12,13,14,15,16\right\rbrace \text{.}
\end{equation*}
Let $C$ be a bordered 4-circulant code of length $n=23$ with 
generator matrix $G_{C}$ in the form given in Equation (\ref{eq:bordered4circ}) 
with $e=\omega$ and with the circulant matrices $A,B$ generated by, respectively,
\begin{equation*}
\begin{aligned}
&(\overline{\omega},1,\overline{\omega},\omega,\omega,1,1,0,0,0,0)\text{, and} \\
&(0,\overline{\omega},\omega,\omega,\omega,\omega,\overline{\omega},1,0,1,\omega) \text{.}
\end{aligned}
\end{equation*}
Use the rows of $G_{C}$ indexed by the set $P$ as the rows of $G_{D}$, 
the generator matrix of a subcode $D$ of $C$. Using $D \subset C$, 
a $[[23,6,8/2]]_4$ asymmetric quantum code $Q$ can be constructed.

If we use the same code $C$ but $G_{D}$ is now $G_{C}$ with rows 3,6,7,9, and 11 
deleted, then, in a similar manner, we get a $[[23,2.5,8/4]]_4$ code $Q$.

\subsection{Construction from two proper subcodes}\label{subsec:proper}
In the previous two subsections, the larger code $C$ is an additive 
self-dual code while the subcode $D$ of $C$ is constructed by 
deleting rows of $G_{C}$. New or better 
asymmetric quantum codes can be constructed from two nested 
proper subcodes of an additive self-dual code. The following 
two examples illustrate this fact.

\begin{ex}\label{ex:9.1}
Let $C$ be a self-dual Type II additive code of length 22 with 
generating vector
\begin{equation*}
\v=(\omega,\omega,\omega,\omega,1,1,\overline{\omega},
\omega,0,\omega,\omega,0,\omega,\omega,1,1,0,\ldots,0)\text{.}
\end{equation*}
Let $G_{C}$ be the generator matrix of $C$ from the cyclic 
development of $\v$. Derive the generator matrices $G_{D}$ 
of $D$ and $G_{E}$ of $E$ by deleting, respectively, 
the last $10$ and $11$ rows of $G_{C}$. Applying Theorem~\ref{thm:3.5} 
on $E \subset D$ yields an asymmetric $[[22,0.5,10/2]]_4$-code $Q$.
\end{ex}

\begin{ex}\label{ex:9.2}
Let $C$ be a self-dual Type I additive code of length 25 labeled 
$C_{25,4}$ in~\cite{GK04} with generating vector
\begin{equation*}
\v=(1,1,\omega,0,1,\overline{\omega},1,0,\omega,1,1,0,0,\ldots,0)\text{.}
\end{equation*}
Let $G_{C}$ be the generator matrix of $C$ from the cyclic 
development of $\v$. Derive the generator matrices $G_{D}$ 
of $D$ and $G_{E}$ of $E$ by deleting, respectively, 
the last $5$ and $6$ rows of $G_{C}$. An asymmetric 
$[[25,0.5,9/4]]_4$-code $Q$ is hence constructed.
\end{ex}

\section{Conclusions and Open Problems}\label{sec:Conclusion}
In this paper, we establish a new method of deriving asymmetric 
quantum codes from additive, not necessarily linear, codes 
over the field $\F_{q}$ with $q$ an even power of a prime $p$.

Many asymmetric quantum codes over $\F_{4}$ are constructed.
These codes are different from those listed in prior works 
(see~\cite[Ch. 17]{Aly08} and~\cite{SRK09}) on asymmetric 
quantum codes.

There are several open directions to pursue. 
On $\F_{4}$-additive codes, exploring the 
notion of nestedness in tandem with the dual distance 
of the inner code is a natural continuation if 
we are to construct better asymmetric quantum codes. 
An immediate project is to understand such relation 
in the class of cyclic (not merely circulant) codes 
studied in~\cite{Huff07}.

Extension to codes over $\F_9$ or $\F_{16}$ is 
another option worth considering. More generally, 
establishing propagation rules may help us find better 
bounds on the parameters of asymmetric quantum codes.

\section*{Acknowledgment}\label{sec:Acknowledge}
The first author thanks Keqin Feng for fruitful discussions.
The authors thank Iliya Bouyukliev, Yuena Ma, and Ruihu Li 
for sharing their data on known Hermitian self-orthogonal $\F_{4}$-codes, 
and Somphong Jitman for helpful discussions.

\begin{IEEEbiographynophoto}{Martianus~Frederic~Ezerman} (S'10) 
	grew up in East Java Indonesia. He received his BA in philosophy and 
BSc in mathematics in 2005 and his MSc in mathematics in 2007, all 
from Ateneo de Manila University, Philippines. 

He is currently a PhD candidate under research scholarship at the 
Division of Mathematical Sciences, School of Physical and Mathematical 
Sciences, Nanyang Technological University, Singapore. His main research 
interest is coding theory, focusing on quantum error-correcting codes.
\end{IEEEbiographynophoto}

\begin{IEEEbiographynophoto}{San~Ling}
received the BA degree in mathematics from the University of Cambridge 
in 1985 and the PhD degree in mathematics from the University of California, 
Berkeley in 1990.

Since April 2005, he has been a Professor with the Division of Mathematical 
Sciences, School of Physical and Mathematical Sciences, in Nanyang 
Technological University, Singapore. Prior to that, he was with the 
Department of Mathematics, National University of Singapore. 

His research fields include arithmetic of modular curves and 
applications of number theory to combinatorial designs, coding theory, 
cryptography and sequences.
\end{IEEEbiographynophoto}

\begin{IEEEbiographynophoto}{Patrick~Sol{\'e}} (M'88) 
received the Ing\'enieur and Docteur-Ing{\'e}nieur degrees both from
Telecom ParisTech, Paris, France, in 1984 and 1987, respectively, and
the habilitation \`a diriger des recherches from Universit\'e de
Nice-Sophia Antipolis, Sophia Antipolis, France, in 1993.

He has held visiting positions in Syracuse University, Syracuse, NY,
from 1987 to 1989, Macquarie University, Sydney, Australia, from 1994
to 1996, and Lille University, Lille, France, from 1999 to 2000. From 
1989 to 2009, he has been a permanent member of the CNRS
Laboratory I3S, Sophia Antipolis, France, and from 2009 to present of
the CNRS Laboratory LTCI, Paris, France.

His research interests include coding theory (covering radius, codes
over rings, geometric codes), interconnection networks (graph spectra,
expanders), vector quantization (lattices), and cryptography (Boolean
functions). Dr. Sol{\'e} is the recipient (jointly with Hammons, Kumar,
Calderbank, and Sloane) of the IEEE Information Theory Society Best
Paper Award in 1994. He has served as an associate editor of the
Transactions from 1999 until 2003.
\end{IEEEbiographynophoto}
\end{document}